\documentclass{article}
\PassOptionsToPackage{table}{xcolor}
\usepackage{arxiv}
\renewcommand{\headeright}{}
\renewcommand{\undertitle}{}

\usepackage[utf8]{inputenc}
\usepackage[T1]{fontenc}
\usepackage{amsmath, amssymb, amsthm}
\usepackage{graphicx}
\usepackage{booktabs}
\usepackage{tabularx}
\usepackage{longtable}
\usepackage{multirow}
\usepackage{array}
\usepackage{colortbl}
\usepackage{tikz}
\usetikzlibrary{positioning, arrows.meta, calc, decorations.pathreplacing, shapes.geometric, backgrounds, fit}
\usepackage{csquotes}
\usepackage[numbers, sort&compress]{natbib}
\usepackage{hyperref}
\hypersetup{colorlinks=true, linkcolor=black, citecolor=black, urlcolor=black}

\definecolor{purple600}{HTML}{534AB7}
\definecolor{teal600}{HTML}{1D9E75}
\definecolor{coral600}{HTML}{D85A30}
\definecolor{amber600}{HTML}{BA7517}
\definecolor{blue600}{HTML}{185FA5}
\definecolor{green600}{HTML}{3B6D11}
\definecolor{gray50}{HTML}{F1EFE8}
\definecolor{gray200}{HTML}{B4B2A9}
\definecolor{gray600}{HTML}{5F5E5A}
\definecolor{red600}{HTML}{C42B1C}
\definecolor{gold600}{HTML}{B8860B}

\newtheorem{definition}{Definition}[section]

\newtheorem{proposition}{Proposition}[section]
\newtheorem{theorem}{Theorem}[section]

\title{The Dignity-Centric Stack:\\[2pt] A Commons-Governed, Horizontally Federated Architecture for Human-Dignity AI}

\author{%
  Eduardo C\'{e}sar Garrido Merch\'{a}n\\
  Department of Quantitative Methods\\
  Universidad Pontificia Comillas (ICADE / IIT), Madrid\\
  \texttt{ecgarrido@comillas.edu}\\
}

\date{2026 \quad|\quad Preprint}

\begin{document}
\maketitle

\begin{abstract}
\noindent
The human-dignity-centric digital social contract grounds personal data in human dignity, data personalism, and data sovereignty, and articulates six dimensions of data governance: technological oversight, automation limits, economic justice, political legitimacy, social cohesion, and legal guarantees. It presupposes, however, that enforcement falls to State regulators, licensed fiduciaries, and multi-stakeholder bodies embedded in existing legal systems. This paper asks whether its normative content can instead be realized not as rules imposed on the owners of the AI stack from without, but as a commons-governed infrastructure that any person, firm, or State may use and fund while its governance stays horizontal, polycentric, and subsidiary. We construct the \emph{Dignity Stack}, a six-layer architecture mapping each dimension onto a layer of commons-governed AI infrastructure, with protocols drawn from the Liberation Stack framework and from the cooperative, mutualist, and libertarian-municipalist traditions. The commons is State-agnostic rather than anti-State, anarchist in its horizontal means but not in the abolition of the State. Its central device is a decoupling of capital from control, by which the stack functions as a shared civic battery, charged by many contributors yet steered by none in proportion to its charge. We prove that this defeats formal capture through votes or surplus, and show that structural capture, the leverage of a dominant supplier free to withdraw what it provides, is resisted only insofar as operational supply is polycentric and substitutable, a condition demanding at the lower layers and perhaps presently unattainable at chip fabrication. We conclude, with explicit attention to its limits, that commons-governed AI realizes the values the contract proclaims more faithfully than the regulation it presupposes.

\vspace{8pt}
\noindent\textbf{Keywords:} digital social contract, data personalism, commons-governed AI, commons-based peer production, polycentric governance, subsidiarity, horizontal governance, capital--governance decoupling, data sovereignty, human dignity, Liberation Stack, Dignity-Centric Stack, convivial tools, mutual aid, federated governance.
\end{abstract}

\section{Introduction: The Enforcement Problem in the Digital Social Contract}
\label{sec:intro}

The proposition that personal data constitutes a fundamental human right, advanced with systematic rigour by Alvarez-Pallete et al.~\citep{alvarezpallete2026}, represents a significant philosophical advance over the prevailing regulatory paradigm, which treats data protection as a consumer right to be balanced against commercial interests. By grounding data governance in the Kantian categorical imperative, in Mounier's personalism, and in the concept of contextual integrity, the digital social contract framework elevates the discussion from regulatory technique to political philosophy: the question is no longer how to constrain data extraction within acceptable limits but whether data extraction as such is compatible with the irreducible dignity of the human person.

Yet the framework, precisely because of its philosophical ambition, generates a tension that its authors do not resolve. The normative content of the contract demands that persons never be treated merely as means to algorithmic ends, that consent be meaningful rather than nominal, that economic value flow back to data subjects, and that governance be participatory rather than delegated. The institutional implementation the contract proposes, however, relies on State regulatory authorities to enforce fiduciary duties, on licensed intermediaries operating within existing legal systems, on multi-stakeholder governance bodies embedded in representative democratic structures, and on courts to provide remedies. The question that this paper takes as its point of departure is whether these institutional forms are adequate to the normative demands they are meant to serve, or whether the very structures the contract relies upon for enforcement are implicated in the violations it seeks to prevent.

The question is not rhetorical. The State that is called upon to enforce data fiduciary duties is the same State whose intelligence agencies conduct mass surveillance, whose patent law enforces the monopolies that concentrate AI infrastructure, whose procurement policies channel public funds toward corporate data centres, and whose export controls weaponize the chip supply chain. The licensed intermediaries that are to serve as data trustees operate within a legal system whose corporate governance norms structurally favour capital over labour and whose fiduciary law has historically served the interests of wealth preservation. The multi-stakeholder governance bodies that are to ensure democratic participation reproduce the same pathologies of captured consultation and manufactured consent that characterize regulatory proceedings in telecommunications, finance, and environmental policy. To note these structural entanglements is not to impugn the motives of the contract's authors but to observe that the institutional mechanisms they propose operate within, and are shaped by, the very power asymmetries they are designed to correct.

This paper develops an alternative institutional architecture for the digital social contract, one that preserves and indeed strengthens its normative commitments while shifting the locus of enforcement from statist regulation imposed from above to commons-governed, self-organizing institutions whose authority is held horizontally by their participants. The alternative is not a thought experiment but a systematic construction built on actually existing organizational models: cooperative data trusts, commons-governed infrastructure, municipal assemblies, mutual-aid networks, and federated protocols, all of which have demonstrated their operational viability at non-trivial scale. The relation to the State is not antagonistic but subsidiary: the commons governs at the level closest to those affected, and the State, like any other actor, may use the infrastructure and contribute to it without thereby acquiring control over it. The central claim is that the normative content of the digital social contract, specifically the irreducibility of dignity, the requirement of meaningful consent, the prohibition on instrumentalization, and the demand for economic justice, is not merely compatible with commons-based institutional forms but is more faithfully realized by them than by the purely regulatory alternatives the original contract presupposes.

The infrastructure over which this contest is waged has by now acquired a definite shape, and one conceded even by those who profit from it. \citet{huang2026}, writing as the chief executive of the industry's dominant accelerator vendor, describes the contemporary artificial intelligence system as a five-layer stack, rising from energy through chips, infrastructure, and models to the applications in which economic value is realized. We invoke this partition not as an independent analysis of the political economy of computing but as the incumbent's own self-description of the structure it dominates. Each of these layers is today overwhelmingly the property of a small number of firms: energy and data centres concentrated in the hands of hyperscalers, chips and their software ecosystem controlled by a single dominant vendor, foundation models trained and gated by a handful of laboratories, and applications mediated by the platforms that own the layers beneath them. The digital social contract is, read in this light, a claim about how the five-layer stack ought to be governed, asking that the persons whose data animates the models and applications be treated as ends rather than as raw material and that the value generated at the upper layers flow back to those who make it possible. What the contract does not supply is an account of ownership and control over the stack itself, presupposing instead that a State external to the stack will regulate its private owners. The architecture we develop supplies the missing account. It is a governance overlay deployed over the same five-layer technical substrate, placing each layer under commons-based, horizontal, polycentric control, so that the normative demands of the digital social contract are met not by regulating the owners of the stack from without but by changing the terms on which the stack is governed from within. The overlay does not require expropriating the incumbents who hold the layers today. It specifies instead a structure into which capital, compute, and data may flow from any source, including the hyperscalers, the dominant accelerator vendor, the laboratories, and the States that subsidize them, while the right to set the system's purposes, norms, and limits remains with its participants on equal terms. The stack so governed is in this sense a shared civic battery rather than a private asset: it can be charged by whoever has the means, while who steers it is held apart from who charges it, to the degree that no single source becomes indispensable. The decoupling of contribution from control, developed formally in Section~\ref{sec:integration}, is the device that makes this combination of open capitalization and horizontal governance coherent rather than contradictory, and we are careful there to mark the conditions under which it holds and those under which it fails. The six layers of the Dignity Stack are therefore not a re-partition of the five technical layers. They track the six normative dimensions of the digital social contract rather than the engineering decomposition of the machine, so that a single governance concern such as economic justice ranges across several technical layers, while a single technical layer such as infrastructure is implicated in several governance concerns.

The argument is structured as follows. Section~\ref{sec:foundations} establishes the philosophical foundations, examining the convergence and divergence between data personalism and the theory of commons-based, polycentric governance on the questions of dignity, autonomy, and subsidiarity, with transparent acknowledgment of which commitments of the cited traditions we draw upon and which we set aside. Section~\ref{sec:architecture} introduces the six-layer Dignity Stack architecture, providing formal definitions for each layer. Sections~\ref{sec:dignity_layer} through~\ref{sec:economic_layer} develop each layer in detail, specifying the voluntary organizational protocol that implements the corresponding dimension of the digital social contract. Section~\ref{sec:integration} presents the integrated architecture and analyzes the inter-layer dependencies that ensure systemic coherence. Section~\ref{sec:objections} addresses eight principal objections, including capture of the commons by its largest funders, data irreversibility in trained models, the tyranny of structurelessness, the tragedy of the anticommons, and the concentration of capital and compute in the lower layers of the AI stack. Section~\ref{sec:conclusion} concludes. Figure~\ref{fig:dignity_stack} provides a visual overview.

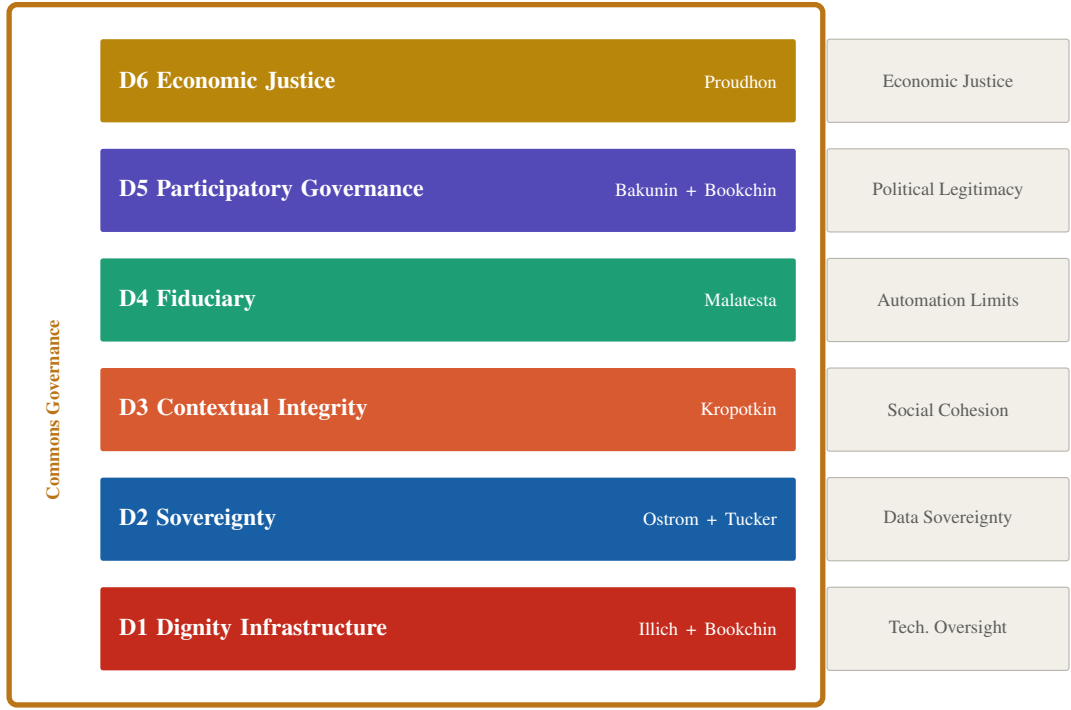
\begin{figure}[t]
\centering
\begin{tikzpicture}[
  layer/.style={
    minimum width=9.2cm,
    text width=8.7cm,
    minimum height=1.1cm,
    text=white,
    font=\small,
    rounded corners=1pt,
    inner sep=4pt
  },
  dsc/.style={
    minimum width=3.2cm,
    minimum height=1.1cm,
    font=\scriptsize,
    rounded corners=1pt,
    inner sep=3pt,
    draw=gray200,
    fill=gray50,
    text=gray600
  }
]
\node[layer, fill=red600] (d1) at (0, 0)
  {\textbf{D1 Dignity Infrastructure}\hfill
   \textrm{\scriptsize Illich + Bookchin}};
\node[layer, fill=blue600] (d2) at (0, 1.45)
  {\textbf{D2 Sovereignty}\hfill
   \textrm{\scriptsize Ostrom + Tucker}};
\node[layer, fill=coral600] (d3) at (0, 2.9)
  {\textbf{D3 Contextual Integrity}\hfill
   \textrm{\scriptsize Kropotkin}};
\node[layer, fill=teal600] (d4) at (0, 4.35)
  {\textbf{D4 Fiduciary}\hfill
   \textrm{\scriptsize Malatesta}};
\node[layer, fill=purple600] (d5) at (0, 5.8)
  {\textbf{D5 Participatory Governance}\hfill
   \textrm{\scriptsize Bakunin + Bookchin}};
\node[layer, fill=gold600] (d6) at (0, 7.25)
  {\textbf{D6 Economic Justice}\hfill
   \textrm{\scriptsize Proudhon}};

\node[dsc, anchor=west] at ([xshift=0.4cm]d1.east) {Tech.\ Oversight};
\node[dsc, anchor=west] at ([xshift=0.4cm]d2.east) {Data Sovereignty};
\node[dsc, anchor=west] at ([xshift=0.4cm]d3.east) {Social Cohesion};
\node[dsc, anchor=west] at ([xshift=0.4cm]d4.east) {Automation Limits};
\node[dsc, anchor=west] at ([xshift=0.4cm]d5.east) {Political Legitimacy};
\node[dsc, anchor=west] at ([xshift=0.4cm]d6.east) {Economic Justice};

\draw[amber600, line width=2pt, rounded corners=3pt]
  ([xshift=-1.2cm, yshift=-0.45cm]d1.south west)
  rectangle
  ([xshift=0.35cm, yshift=0.45cm]d6.north east);
\node[rotate=90, font=\scriptsize\bfseries, text=amber600, anchor=center]
  at ([xshift=-0.62cm]d3.west) {Commons Governance};

\end{tikzpicture}
\caption{The Dignity Stack: six-layer governance architecture implementing the digital social contract through commons-based organizational protocols. Each layer (left) maps to a dimension of the original DSC framework (right). The amber frame represents the overarching constitutional principle of horizontal, commons-based governance that binds all layers.}
\label{fig:dignity_stack}
\end{figure}

\section{Philosophical Foundations: Data Personalism and the Theory of the Commons}
\label{sec:foundations}

This section examines the three philosophical pillars on which the digital social contract rests and argues that each generates institutional implications that point beyond the purely statist framework the contract presupposes, toward commons-based, polycentric, and subsidiary alternatives. The argument is not that these implications are uniquely derivable from the cited philosophical traditions, since each tradition admits multiple legitimate readings, but that the commons-based reading we develop is at least as faithful to their core commitments as the purely regulatory reading the contract assumes, and in several respects more so.

A preliminary clarification is essential. We do not claim that Kant, Mounier, or Nissenbaum were anarchists, or that their philosophical frameworks entail anarchism as a matter of strict logical necessity. Kant, in the \emph{Metaphysics of Morals}~\citep{kant1797}, explicitly defends the State's right to coerce as a condition for the external realization of freedom, arguing that ``hindering a hindrance to freedom is consistent with freedom according to universal laws'' (6:231). Mounier, while deeply critical of the centralized bureaucratic State, explicitly rejected anarchism, holding that the common good requires institutional constraints on individual action~\citep{mounier1938}. Nissenbaum developed contextual integrity as a framework for regulatory policy, not as an argument against regulation~\citep{nissenbaum2004}. We acknowledge these positions and do not misrepresent them. What we argue is that the digital social contract selectively mobilizes these traditions, invoking their critiques of instrumentalization, commodification, and contextual boundary violation, while relying on institutional mechanisms that are themselves subject to the same critiques. Our reading is equally selective, but in a direction that resolves this tension rather than reproducing it: we adopt the horizontal, federative, and mutualist organizational forms developed most fully within the cooperative and anarchist traditions, while declining their rejection of the State as such. The commons we describe is not an anti-State project but a State-agnostic one. It neither requires the abolition of public authority nor depends upon it for its internal governance; it can be used and funded by States and firms while remaining governed horizontally by its participants. In this our proposal stands closer to Mounier's own pluralism, which sought subsidiary institutions mediating between the person and the State, than to the anarchism from which we borrow organizational technique. The label we retain is therefore commons-based and subsidiary rather than anarchist: the architecture is anarchist in its means, the horizontal and voluntary forms of self-organization, but not in its ends, the abolition of the State. This State-agnostic posture occupies a definite position rather than an unstable midpoint. Where a State or firm merely participates in and funds the commons, the commons governs only its own internal affairs and binds only its voluntary participants, which requires no authority over anyone. Where the commons constrains the data practices of a powerful participant, it does so on the strength of that participant's voluntary accession together with the minimal legal recognition of the right of associations to govern their own affairs, which is Ostrom's seventh design principle, and not on any claim to coerce the unwilling. We do not claim the commons can govern a State that refuses such recognition; against a State determined to override it the commons is, like every private owner of the same infrastructure, without recourse. The proposal is in this sense a recognized, subsidiary commons rather than either a delegated public regulator or an autonomous counter-power, and we locate it there deliberately rather than leaving it to hover between the two.

The first pillar is the Kantian prohibition on treating persons merely as means. Alvarez-Pallete et al.\ invoke Kant's second formulation of the categorical imperative to ground their distinction between DatAism, the paradigm that reduces persons to data points for algorithmic optimization, and HumAism, the counter-paradigm that places human dignity at the centre of data governance~\citep{alvarezpallete2026}. The invocation is philosophically sound: surveillance capitalism, as analyzed by Zuboff~\citep{zuboff2019}, systematically treats persons as raw material for prediction markets, violating the categorical imperative at the level of institutional design rather than individual action. The critical word in Kant's formula is ``merely'': persons must be treated ``never merely as a means, but always at the same time as an end.'' Kant himself held that the State can coerce while treating persons as ends, because coercion in the service of rightful freedom respects rational agency by maintaining the conditions under which all persons can exercise their freedom~\citep{kant1797}. A long line of critics, from William Godwin's \emph{Enquiry Concerning Political Justice}~\citep{godwin1793} through Kropotkin's mutual aid~\citep{kropotkin1902}, contests this Kantian reconciliation on empirical grounds: whatever the theoretical possibility of a State that coerces only in the service of freedom, actually existing States routinely coerce in the service of corporate interests, class preservation, and self-perpetuation. The question for data governance is not whether an ideal State could enforce dignity, but whether the States that actually exist do so, or whether their structural entanglement with the power asymmetries that generate data extraction makes them unreliable guarantors of the values they proclaim. We do not draw from this the anarchist conclusion that the State must be abolished. We draw the weaker and better-supported claim that an architecture which makes dignity depend on the goodwill of any single authority, whether a public regulator or a private owner of the infrastructure, is more fragile than one in which governance is distributed, consent-based, and structurally hard to capture. The commons is the constructive response to the empirical gap between Kant's theory of the rightful State and the reality of both State and corporate complicity in data extraction. We grant at once that distributing governance does not by itself defeat the resource asymmetry that drives capture, and that the commons is, as Section~\ref{sec:integration} concedes, vulnerable to the same asymmetry rather than immune to it. The inference is not that distribution abolishes capture but the narrower one that polycentric, consent-based governance localizes the consequences of any single capture and renders capture detectable and contestable, where a unitary authority both concentrates the stakes of capture and, being the arbiter of its own conduct, obscures it. This lowers the cost of the empirical fact that authorities are captured without pretending to eliminate the fact.

\begin{definition}[Non-instrumentalization thesis]
\label{ax:noninst}
A governance architecture for data more faithfully respects the Kantian prohibition on instrumentalization to the extent that it minimizes the number of institutional relationships in which persons are subject to the will of others without their ongoing, meaningful, and revocable consent. An architecture constituted entirely through voluntary association satisfies this condition more fully than one that includes institutions whose authority does not depend on the consent of those subject to it.
\end{definition}

The second pillar is Mounier's personalism, which the digital social contract mobilizes to argue that data is not a detachable commodity but an expression of irreducible personal identity~\citep{mounier1938}. Mounier's contribution is to insist that the person is constituted not in isolation but in relation: the person exists as a being-in-community, and dignity is realized through free engagement with others, not through atomized autonomy. This relational conception of the person is the foundation of data personalism, the thesis that data governance must respect both individual autonomy and the relational context in which data is generated and interpreted. Mounier distinguished the person from the individual, where the individual is the liberal-bourgeois abstraction of the isolated rights-bearer and the person is the concrete being embedded in community, and he argued that both the capitalist market and the centralized bureaucratic State reduce persons to individuals by dissolving the organic communities through which personhood is realized~\citep{mounier1938}. Mounier did not, however, conclude from this critique that the State should be abolished; he sought a reformed, decentralized, pluralist political order in which subsidiary institutions mediate between the person and the State, and in which the State retains a positive role, protecting and fostering those intermediary bodies rather than dissolving or supplanting them. We do not follow the anarchist communitarians who take the critique further, holding that voluntary communities should replace the State outright. On the governance form our reading is genuinely closer to Mounier's pluralism than to that anarchism: the commons is the subsidiary institution, governing the stewardship of personal data at the level closest to those whose data it is, neither absorbed into the State nor pitted against it, and compatible with a State that recognizes and protects it in the manner Mounier envisaged. We are careful, however, not to claim Mounier's endorsement beyond what the texts support, and two qualifications are required. First, Mounier's subsidiarity assigns the State a protective function over intermediary bodies; our architecture accommodates this through the recognition of the right to organize rather than through opposition to the State, but it does not depend on such protection, and to that extent it asks of subsidiarity slightly less vertical structure than Mounier did. Second, and more substantially, Mounier accepted the market as a legitimate means of meeting material needs, criticizing its dominance rather than its existence, whereas the economic layer we develop in Section~\ref{sec:economic_layer} adopts a Proudhonian mutualism that withdraws the residual claim of capital. This is our extension, not Mounier's position, and we mark it as such: the Dignity Stack draws on Mounier's subsidiary mediation of personhood while going beyond his economics, and our claim of fidelity is to the former and not the latter.

\begin{definition}[Relational personhood thesis]
\label{ax:relational}
The human person is constituted in and through relations with others. A data governance architecture more faithfully realizes Mounerian personalism to the extent that it preserves and strengthens the voluntary relational contexts through which personhood is realized, rather than substituting bureaucratic or algorithmic intermediation for direct human relations of mutual recognition.
\end{definition}

The third pillar is contextual integrity, the principle, developed by Nissenbaum~\citep{nissenbaum2004} and mobilized in the digital social contract via Cohen and Solove, that the appropriateness of data flows depends on the social context in which data is generated and the norms governing that context. Health data collected for treatment is governed by different norms than commercial data collected for advertising; violations of contextual integrity occur when data flows cross contextual boundaries without respecting the norms of either context. The digital social contract proposes to enforce contextual integrity through regulatory mechanisms: purpose limitation requirements, mandatory impact assessments, and legal remedies for boundary violations. The difficulty is that the State, as a unitary sovereign, is itself a persistent violator of contextual integrity: intelligence agencies routinely access data collected in commercial contexts for surveillance purposes, tax authorities access financial data for enforcement purposes unrelated to the context of its collection, and law enforcement agencies increasingly demand access to health, communications, and location data that persons disclosed in contexts governed by norms of confidentiality. An enforcement mechanism that is also a principal source of the violations it is meant to prevent faces a structural credibility problem. Nissenbaum herself does not draw this conclusion; her framework is designed for regulatory application. We extend her framework by observing that a governance architecture in which no single authority has jurisdiction across all contexts embodies contextual integrity architecturally, making boundary violations structurally difficult rather than merely illegal. This is not a derivation from Nissenbaum but a constructive proposal that takes her normative framework and asks what institutional form would most fully realize it.

\begin{definition}[Structural contextual integrity thesis]
\label{ax:context}
A data governance architecture embodies contextual integrity more robustly to the extent that its jurisdictional structure is context-specific rather than territorially unitary. An architecture composed of context-specific governance units with no single authority spanning all contexts makes contextual boundary violations structurally difficult; an architecture with a sovereign authority that can override contextual norms by fiat makes them structurally easy, relying on regulatory constraint rather than architectural design to prevent them.
\end{definition}

These three theses establish the normative foundations from which the Dignity Stack is constructed. They are not axioms in the sense of self-evident truths, nor are they strict derivations from Kant, Mounier, or Nissenbaum. They are substantive philosophical positions, grounded in but extending those traditions, that we defend on the following basis: (i) the non-instrumentalization thesis captures the empirical reality that actually existing States are unreliable guarantors of the dignity they theoretically protect; (ii) the relational personhood thesis extends Mounier's critique of the centralized State to its logical institutional conclusion; and (iii) the structural contextual integrity thesis takes Nissenbaum's normative framework and asks what architecture would realize it most fully. Together, they motivate a governance architecture that is voluntary in its constitution, relational in its organizational form, and pluralist in its jurisdictional structure.

\begin{proposition}[Voluntary governance as more faithful realization]
\label{prop:entailment}
A governance architecture in which authority over each context rests on the ongoing consent of its participants, with polycentric jurisdiction and no single authority spanning all contexts, satisfies the non-instrumentalization thesis, the relational personhood thesis, and the structural contextual integrity thesis more fully than any architecture that vests governance in an authority whose standing does not depend on that consent, whether a coercive State regulator or a controlling owner of capital, provided the commons-based architecture is operationally viable.
\end{proposition}

The argument is straightforward. A voluntary architecture eliminates the institutional relationships in which persons are subject to authority without ongoing consent, satisfying the non-instrumentalization thesis by construction. It preserves voluntary relational contexts rather than displacing them with bureaucratic intermediation, satisfying the relational personhood thesis. And its polycentric jurisdictional structure makes contextual boundary violations structurally difficult, satisfying the structural contextual integrity thesis. The critical proviso is operational viability: a voluntary architecture that is theoretically preferable but practically impossible would not constitute a more faithful realization. The remainder of this paper is devoted to demonstrating that the Dignity Stack is operationally viable, drawing on actually existing institutional models at each layer. The question is therefore not whether commons-based governance is normatively preferable in the abstract, but whether it can be constructed concretely for the domain of AI data governance.

One limitation of this argument must be conceded at the outset, since it qualifies the force of every claim of more faithful realization that follows. The three theses operationalize the values of the digital social contract, non-instrumentalization, relational personhood, and contextual integrity, in terms that a voluntary, polycentric architecture satisfies comparatively well: the minimization of non-consensual institutional relationships, the preservation of voluntary relational contexts, and the dispersal of jurisdiction. A defender of statist regulation may reasonably contest these operationalizations, holding for instance that the Kantian prohibition on instrumentalization is honored not by minimizing coercive relationships but by subjecting them to public, rule-of-law constraint, which is in fact closer to Kant's own position at 6:231. We do not claim to have refuted that reading. We claim that our operationalizations are defensible articulations of the same values, at least as faithful to the cited traditions as the statist reading, and that on these articulations the commons architecture scores higher. The comparative is in this sense conditional on the operationalization, and we advance it as a defended position rather than a demonstration. To a reader who rejects the operationalization, the contribution of the paper is the weaker but still substantive one of exhibiting a concrete, commons-based architecture that satisfies the dignity-centred values under a natural reading of them.

\section{The Dignity Stack: Architecture and Formal Specification}
\label{sec:architecture}

The Dignity Stack is a six-layer governance architecture that maps the six dimensions of the digital social contract onto six layers of commons-governed AI infrastructure, assigning to each layer a commons-based organizational protocol that satisfies the normative requirements of the corresponding dimension. The community AI infrastructure so governed is the commons-owned counterpart of the centrally held five-layer technical stack described by \citet{huang2026}; the Dignity Stack is a governance overlay over that substrate rather than a technical stack in its own right, and it determines not how the machine is built but who owns and controls it. Before developing each layer, we establish the formal structure.

\begin{definition}[Dignity Stack]
\label{def:stack}
A \emph{Dignity Stack} is a tuple $\mathcal{D} = (D_1, \ldots, D_6, \Phi)$ where each $D_i$ is a governance layer consisting of a triple $(N_i, O_i, V_i)$: a finite set of normative requirements $N_i = \{n_{i1}, \ldots, n_{ik_i}\}$ derived from the corresponding dimension of the digital social contract, an organizational protocol $O_i$ drawn from commons-based institutional models, and a verification predicate $V_i$ defined as the conjunction $V_i(O_i, N_i) = \bigwedge_{j=1}^{k_i} v_{ij}(O_i, n_{ij})$, where each component $v_{ij}$ returns $1$ if and only if $O_i$ provides an institutional mechanism that meets the requirement $n_{ij}$. The relation $\Phi: \{D_1, \ldots, D_6\}^2 \to \{0,1\}$ records the inter-layer dependencies, with $\Phi(D_i, D_j) = 1$ meaning that $D_j$ presupposes $D_i$ and $\Phi(D_i, D_j) = 0$ for every ordered pair not enumerated in Definition~\ref{def:dependency}.
\end{definition}

\begin{definition}[Normative compliance]
\label{def:compliance}
A Dignity Stack $\mathcal{D}$ is \emph{normatively compliant} with the digital social contract if and only if, for every layer $D_i = (N_i, O_i, V_i)$, the verification predicate $V_i(O_i, N_i) = 1$, and the inter-layer dependency function $\Phi$ is satisfied for all dependent pairs.
\end{definition}

A word on the status of this formalism is owed, since it would be a mistake to read the definitions and propositions that follow as theorems in a deductive calculus. The tuple notation regiments three things: the requirements each layer must meet, the protocol proposed to meet them, and the dependency structure that orders the layers. It does not mechanize the verification predicates. Each component $v_{ij}$ encodes a substantive judgement, that a given commons mechanism meets a given normative requirement, and that judgement is discharged not by computation but by the layer-by-layer arguments of Sections~\ref{sec:dignity_layer} through~\ref{sec:economic_layer}. Those arguments, labelled propositions for reference, are arguments in the ordinary philosophical sense that a protocol satisfies a requirement, and they carry exactly the weight such arguments carry, no more. A Dignity Stack that is normatively compliant in the sense of Definition~\ref{def:compliance} is what we mean, throughout, by a stack that \emph{implements} the digital social contract; the phrase has no content beyond the conjunction of the layer verifications and the satisfaction of the dependency relation. The formalism earns its keep by making the requirements, the dependencies, and the points of possible failure explicit and checkable; it does not, and is not meant to, convert contestable normative judgements into derivations. Table~\ref{tab:mapping} summarizes the mapping between the six dimensions of the digital social contract and the six layers of the Dignity Stack, together with the commons-based organizational protocol assigned to each layer.

\begin{table}[t]
\centering
\footnotesize
\setlength{\tabcolsep}{4pt}
\renewcommand{\arraystretch}{1.25}
\rowcolors{2}{gray50}{white}
\begin{tabularx}{\linewidth}{>{\bfseries}lllX}
\toprule
\rowcolor{purple600}
\textcolor{white}{Layer} &
\textcolor{white}{DSC Dimension} &
\textcolor{white}{Commons Protocol} &
\textcolor{white}{Core mechanism} \\
\midrule
D1 & Technological Oversight & Illich + Bookchin & Convivial thresholds; municipal assembly \\
D2 & Data Sovereignty & Ostrom + Tucker & Cooperative data trusts; open protocols \\
D3 & Social Cohesion & Kropotkin & Federated mutual-aid data sharing \\
D4 & Automation Limits & Malatesta & Voluntary fiduciary commitment; exit rights \\
D5 & Political Legitimacy & Bakunin + Bookchin & Nested assemblies; bottom-up federation \\
D6 & Economic Justice & Proudhon & Mutualist exchange; mutual credit \\
\bottomrule
\end{tabularx}
\caption{Mapping the six dimensions of the digital social contract (Alvarez-Pallete et al., 2026) onto the six layers of the Dignity Stack, with the commons-based organizational protocol assigned to each layer.}
\label{tab:mapping}
\end{table}

The ordering of the layers is not arbitrary. It follows the same bottom-up logic as the Liberation Stack framework: the lower layers provide the material and institutional preconditions for the higher layers, so that D1 (dignity infrastructure) is foundational and D6 (economic justice) presupposes the functioning of all layers beneath it. The inter-layer dependencies specified by $\Phi$ formalize this ordering.

Figure~\ref{fig:taxonomy} presents a comprehensive conceptual taxonomy of the Dignity Stack, tracing the flow from philosophical traditions through normative theses and institutional layers to the concrete enforcement mechanisms and their acknowledged limitations. The taxonomy makes explicit three features of the architecture that the layer-by-layer exposition in Sections~\ref{sec:dignity_layer}--\ref{sec:economic_layer} develops in detail: (i) the philosophical traditions are mobilized selectively, with transparent acknowledgment of where we extend them beyond their authors' own conclusions; (ii) each layer is governed by a specific enforcement mechanism that operates without coercion; and (iii) each mechanism carries a recognized limitation that constrains its effectiveness under specific conditions. The taxonomy thus serves as a map of both the architecture's strengths and its honest boundaries.

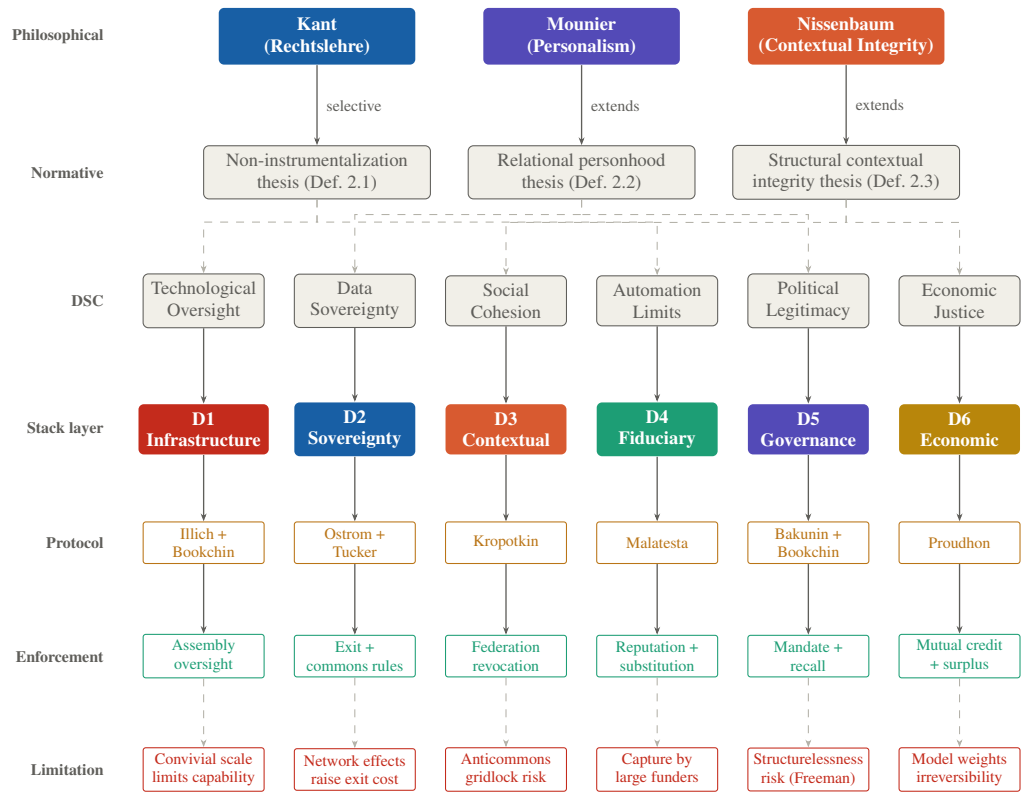
\begin{figure}[p]
\centering
\begin{tikzpicture}[
  every node/.style={font=\scriptsize, align=center},
  trad/.style={
    minimum width=2.6cm, minimum height=0.75cm,
    rounded corners=2pt, font=\scriptsize\bfseries,
    text=white, inner sep=3pt
  },
  thesis/.style={
    minimum width=3.0cm, minimum height=0.65cm,
    rounded corners=2pt, font=\scriptsize,
    draw=gray600, fill=gray50, text=gray600, inner sep=3pt
  },
  layer/.style={
    minimum width=2.3cm, minimum height=0.65cm,
    rounded corners=2pt, font=\scriptsize\bfseries,
    text=white, inner sep=3pt
  },
  thinker/.style={
    minimum width=2.3cm, minimum height=0.55cm,
    rounded corners=1pt, font=\tiny,
    draw=amber600, fill=white, text=amber600, inner sep=2pt
  },
  mech/.style={
    minimum width=2.3cm, minimum height=0.55cm,
    rounded corners=1pt, font=\tiny,
    draw=teal600, fill=white, text=teal600, inner sep=2pt
  },
  limit/.style={
    minimum width=2.3cm, minimum height=0.55cm,
    rounded corners=1pt, font=\tiny,
    draw=red600, fill=white, text=red600, inner sep=2pt
  },
  arr/.style={-{Stealth[length=3pt]}, gray600, line width=0.5pt},
  darr/.style={-{Stealth[length=3pt]}, gray200, line width=0.4pt, dashed}
]

\node[font=\tiny\bfseries\color{gray600}, anchor=east] at (-6.2, 0) {Philosophical};
\node[font=\tiny\bfseries\color{gray600}, anchor=east] at (-6.2, -1.8) {Normative};
\node[font=\tiny\bfseries\color{gray600}, anchor=east] at (-6.2, -3.5) {DSC};
\node[font=\tiny\bfseries\color{gray600}, anchor=east] at (-6.2, -5.2) {Stack layer};
\node[font=\tiny\bfseries\color{gray600}, anchor=east] at (-6.2, -6.7) {Protocol};
\node[font=\tiny\bfseries\color{gray600}, anchor=east] at (-6.2, -8.2) {Enforcement};
\node[font=\tiny\bfseries\color{gray600}, anchor=east] at (-6.2, -9.7) {Limitation};

\node[trad, fill=blue600] (kant) at (-3.5, 0) {Kant\\(Rechtslehre)};
\node[trad, fill=purple600] (mounier) at (0, 0) {Mounier\\(Personalism)};
\node[trad, fill=coral600] (nissen) at (3.5, 0) {Nissenbaum\\(Contextual Integrity)};

\node[thesis] (t1) at (-3.5, -1.8) {Non-instrumentalization\\thesis (Def.~2.1)};
\node[thesis] (t2) at (0, -1.8) {Relational personhood\\thesis (Def.~2.2)};
\node[thesis] (t3) at (3.5, -1.8) {Structural contextual\\integrity thesis (Def.~2.3)};

\draw[arr] (kant) -- (t1) node[midway, right, font=\tiny\color{gray600}] {selective};
\draw[arr] (mounier) -- (t2) node[midway, right, font=\tiny\color{gray600}] {extends};
\draw[arr] (nissen) -- (t3) node[midway, right, font=\tiny\color{gray600}] {extends};

\node[thesis, minimum width=1.6cm] (dsc1) at (-5.0, -3.5) {Technological\\Oversight};
\node[thesis, minimum width=1.6cm] (dsc2) at (-3.0, -3.5) {Data\\Sovereignty};
\node[thesis, minimum width=1.6cm] (dsc3) at (-1.0, -3.5) {Social\\Cohesion};
\node[thesis, minimum width=1.6cm] (dsc4) at (1.0, -3.5) {Automation\\Limits};
\node[thesis, minimum width=1.6cm] (dsc5) at (3.0, -3.5) {Political\\Legitimacy};
\node[thesis, minimum width=1.6cm] (dsc6) at (5.0, -3.5) {Economic\\Justice};

\draw[darr] (t1.south) -- ++(0,-0.3) -| (dsc1.north);
\draw[darr] (t1.south) -- ++(0,-0.3) -| (dsc4.north);
\draw[darr] (t2.south) -- ++(0,-0.2) -| (dsc2.north);
\draw[darr] (t2.south) -- ++(0,-0.2) -| (dsc5.north);
\draw[darr] (t3.south) -- ++(0,-0.3) -| (dsc3.north);
\draw[darr] (t3.south) -- ++(0,-0.3) -| (dsc6.north);

\node[layer, fill=red600, minimum width=1.6cm] (d1) at (-5.0, -5.2) {D1\\Infrastructure};
\node[layer, fill=blue600, minimum width=1.6cm] (d2) at (-3.0, -5.2) {D2\\Sovereignty};
\node[layer, fill=coral600, minimum width=1.6cm] (d3) at (-1.0, -5.2) {D3\\Contextual};
\node[layer, fill=teal600, minimum width=1.6cm] (d4) at (1.0, -5.2) {D4\\Fiduciary};
\node[layer, fill=purple600, minimum width=1.6cm] (d5) at (3.0, -5.2) {D5\\Governance};
\node[layer, fill=gold600, minimum width=1.6cm] (d6) at (5.0, -5.2) {D6\\Economic};

\foreach \i in {1,...,6} {
  \draw[arr] (dsc\i) -- (d\i);
}

\node[thinker, minimum width=1.6cm] (p1) at (-5.0, -6.7) {Illich +\\Bookchin};
\node[thinker, minimum width=1.6cm] (p2) at (-3.0, -6.7) {Ostrom +\\Tucker};
\node[thinker, minimum width=1.6cm] (p3) at (-1.0, -6.7) {Kropotkin};
\node[thinker, minimum width=1.6cm] (p4) at (1.0, -6.7) {Malatesta};
\node[thinker, minimum width=1.6cm] (p5) at (3.0, -6.7) {Bakunin +\\Bookchin};
\node[thinker, minimum width=1.6cm] (p6) at (5.0, -6.7) {Proudhon};

\foreach \i in {1,...,6} {
  \draw[arr] (d\i) -- (p\i);
}

\node[mech, minimum width=1.6cm] (e1) at (-5.0, -8.2) {Assembly\\oversight};
\node[mech, minimum width=1.6cm] (e2) at (-3.0, -8.2) {Exit +\\commons rules};
\node[mech, minimum width=1.6cm] (e3) at (-1.0, -8.2) {Federation\\revocation};
\node[mech, minimum width=1.6cm] (e4) at (1.0, -8.2) {Reputation +\\substitution};
\node[mech, minimum width=1.6cm] (e5) at (3.0, -8.2) {Mandate +\\recall};
\node[mech, minimum width=1.6cm] (e6) at (5.0, -8.2) {Mutual credit\\+ surplus};

\foreach \i in {1,...,6} {
  \draw[arr] (p\i) -- (e\i);
}

\node[limit, minimum width=1.6cm] (l1) at (-5.0, -9.7) {Convivial scale\\limits capability};
\node[limit, minimum width=1.6cm] (l2) at (-3.0, -9.7) {Network effects\\raise exit cost};
\node[limit, minimum width=1.6cm] (l3) at (-1.0, -9.7) {Anticommons\\gridlock risk};
\node[limit, minimum width=1.6cm] (l4) at (1.0, -9.7) {Capture by\\large funders};
\node[limit, minimum width=1.6cm] (l5) at (3.0, -9.7) {Structurelessness\\risk (Freeman)};
\node[limit, minimum width=1.6cm] (l6) at (5.0, -9.7) {Model weights\\irreversibility};

\foreach \i in {1,...,6} {
  \draw[darr] (e\i) -- (l\i);
}


\end{tikzpicture}
\caption{Conceptual taxonomy of the Dignity Stack. The architecture flows from three philosophical traditions (top), through three normative theses that selectively draw on them, to six DSC dimensions mapped onto six governance layers. Each layer is assigned a commons-based organizational protocol and a concrete enforcement mechanism. The bottom row acknowledges the specific limitation each mechanism faces, ensuring the taxonomy maps both strengths and honest boundaries.}
\label{fig:taxonomy}
\end{figure}

\section{D1. Dignity Infrastructure: Convivial Thresholds and Democratic Oversight}
\label{sec:dignity_layer}

The first dimension of the digital social contract is technological oversight: the requirement that AI systems be designed and deployed in accordance with dignity-by-design principles, including respect for autonomy, transparency, non-manipulation, proportionality, and human oversight~\citep{alvarezpallete2026}. The dignity-by-design checklist proposed by the contract specifies six criteria: genuinely free choice, comprehensible disclosure, data minimization, absence of dark patterns, non-discrimination, and respect for irreducible personhood. The enforcement mechanism proposed is regulatory: designers must document their compliance, and external auditors verify it.

The commons-based alternative begins with a structural observation: a system whose energy and computational scale exceeds the governance capacity of the community it serves cannot satisfy dignity-by-design requirements regardless of how many checklists it passes, because the institutional precondition of meaningful oversight, namely that the community can understand, modify, and if necessary dismantle the system, is absent. This is Illich's convivial threshold applied to AI governance~\citep{illich1973,illich1974}: beyond a certain scale, the tool dominates the user rather than serving the user, and no amount of procedural compliance can substitute for the structural condition of community control.

\begin{definition}[Convivial AI threshold]
\label{def:convivial}
An AI system $A$ operating within a community $C$ satisfies the \emph{convivial threshold} if and only if:
\begin{align}
&\text{(i)}\quad E(A) \leq G(C), \label{eq:energy} \\
&\text{(ii)}\quad K(A) \leq U(C), \label{eq:knowledge} \\
&\text{(iii)}\quad T_{\text{dismantle}}(A) \leq T_{\text{assembly}}(C), \label{eq:dismantle}
\end{align}
where $E(A)$ is the energy consumption of $A$, $G(C)$ is the energy that $C$ can self-provision; $K(A)$ is the technical knowledge required to understand, audit, and modify $A$, and $U(C)$ is the collective technical capacity of $C$; and $T_{\text{dismantle}}(A)$ is the time required to fully decommission $A$, and $T_{\text{assembly}}(C)$ is the time required for $C$ to convene a deliberative assembly and reach a collective decision.
\end{definition}

The three relations are qualitative thresholds rather than numerical optima, and each is read on a commensurable scale. The energy relation~\eqref{eq:energy} compares power in common units. The knowledge relation~\eqref{eq:knowledge} is to be read as set inclusion, the competencies required to understand, audit, and modify $A$ being contained in the competencies the community collectively possesses, rather than as a comparison of scalar quantities. The time relation~\eqref{eq:dismantle} compares durations. A system satisfies the threshold when all three hold and fails it when any one is violated; the predicate is therefore decidable in principle given an inventory of the community's energy, competencies, and deliberative cadence, even where exact measurement is impractical. Condition~\eqref{eq:energy} ensures material sovereignty: the community is not dependent on external energy sources whose providers could exercise leverage over the AI system's operation. Condition~\eqref{eq:knowledge} ensures epistemic sovereignty: the community can audit the system without relying on external experts whose interests may diverge from its own. Condition~\eqref{eq:dismantle} ensures political sovereignty: the community can shut down the system faster than it can deliberate about whether to do so, ensuring that the system never outpaces the community's governance capacity. Together, these three conditions constitute the structural precondition for meaningful technological oversight: a community that can power, understand, and dismantle its AI system has genuine control over it, whereas a community that depends on external energy, external expertise, or external infrastructure for any of these functions has nominal oversight but not real sovereignty.

The organizational protocol for D1 is drawn from the combination of Illich and Bookchin that governs L1 of the Liberation Stack. The community AI system is powered by community-owned renewable energy, governed by a municipal or cooperative assembly, and sized to satisfy the convivial threshold. The dignity-by-design checklist of the digital social contract is not rejected but reconceived: instead of a compliance exercise audited by external regulators, it becomes a set of design constraints enforced by the assembly itself, which has the knowledge and the power to verify compliance directly because the system operates within its epistemic and material sovereignty.

\begin{proposition}[D1 satisfies technological oversight]
\label{prop:d1}
A community AI system that satisfies the convivial threshold (Definition~\ref{def:convivial}) and is governed by a deliberative assembly with binding authority over its design, deployment, and decommissioning satisfies the technological oversight requirements of the digital social contract, including all six dignity-by-design criteria, without recourse to external regulatory enforcement.
\end{proposition}

\begin{proof}
Consider each criterion in turn. Genuinely free choice: the assembly governs by consensus or qualified majority, and any member may propose modification or decommissioning; the convivial threshold ensures that the system can be dismantled faster than the assembly can deliberate, so the community is never locked in. Comprehensible disclosure: condition~\eqref{eq:knowledge} ensures that the community possesses the collective expertise to understand the system's data flows. Data minimization: the assembly defines the purposes for which data is collected and can enforce minimization directly through its binding authority over system design. Absence of dark patterns: because the assembly governs the system's design and its members are the users who would be targeted, manipulative interface patterns are subject to direct audit and challenge by the affected parties, and the assembly can mandate their removal. We do not claim that direct democracy renders dark patterns impossible, since an assembly may authorize, or simply fail to detect, the manipulation of its own members; we claim that design authority lies with the affected parties rather than with a vendor whose revenue is served by manipulation, which removes the structural incentive that produces dark patterns at scale. Non-discrimination: the assembly includes all community members and governs the system's design criteria, ensuring that discriminatory features are subject to direct challenge by those affected. Respect for irreducible personhood: the convivial threshold ensures that no person is rendered dependent on a system that exceeds their community's governance capacity, preserving the structural conditions for autonomous self-determination.
\end{proof}

\section{D2. Sovereignty: Cooperative Data Trusts as Ostromian Commons}
\label{sec:sovereignty_layer}

The second dimension of the digital social contract is data sovereignty: the right of individuals to determine how their data is used, combined with collective sovereignty over data ecosystems, the capacity to withdraw consent, and meaningful choice that transcends the fiction of informed consent through adhesion contracts~\citep{alvarezpallete2026}. The contract proposes data trusts, cooperative governance structures, and interoperability standards as mechanisms for realizing sovereignty. The institutional form it envisions, however, is a licensed intermediary operating within existing legal systems: a data trust chartered under Anglo-American trust law, governed by boards that include data subjects alongside experts and civil society representatives, and regulated by data protection authorities with enforcement powers.

The commons-based alternative retains the data trust as an organizational form but reconceives its institutional character. Instead of a licensed intermediary embedded in State legal systems, the cooperative data trust is a voluntary association constituted by its members, governed as an Ostromian commons, and accountable to its members through exit rights rather than regulatory oversight. The shift is not merely terminological: it transforms the data trust from an institution that governs on behalf of data subjects to one that is governed by them.

\begin{definition}[Cooperative data trust]
\label{def:trust}
An \emph{cooperative data trust} is a voluntary association $\mathcal{T} = (M, R, P, E)$ where $M$ is the set of members (data subjects who have voluntarily joined), $R$ is the set of governance rules adopted by collective decision of $M$, $P$ is the set of data stewardship protocols that govern how members' data is stored, processed, and shared, and $E: M \to \{0,1\}$ is the exit function that permits any member $m \in M$ to withdraw from the trust at any time, taking with them a complete, machine-readable copy of all data they contributed, with the trust retaining no copy after exit.
\end{definition}

The exit function $E$ is the fundamental institutional innovation that distinguishes the cooperative data trust from its statist counterpart. In a State-chartered data trust, the trustee's fiduciary duties are enforced by courts; a member who believes the trust has violated its duties must seek legal remedy, a process that is slow, expensive, and structurally favours the trust's legal resources over the individual's. In the cooperative data trust, enforcement is achieved through exit: a member who judges the trust to have violated its commitments simply leaves, and the trust loses both the member's data and the member's participation. The exit function is not merely a right of withdrawal but a structural mechanism that disciplines the trust's governance: a trust that systematically violates its commitments will haemorrhage members, losing both its data assets and its collective bargaining power, and will ultimately dissolve. This is the commons-based counterpart to what Hirschman~\citep{hirschman1970} terms the ``exit'' mechanism: the capacity to leave an institution without permission disciplines the institution more effectively than the capacity to voice dissent within it, precisely because exit imposes immediate material consequences that voice does not.

The governance rules $R$ are determined by Ostrom's eight design principles for viable commons~\citep{ostrom1990}, adapted to the data trust context. Clearly defined boundaries specify who is a member and what data falls within the trust's stewardship. Congruence between rules and local conditions ensures that governance rules are adapted to the specific community's data practices and values. Collective-choice arrangements guarantee that members participate in modifying governance rules. Monitoring ensures that data stewardship protocols are transparently observed. Graduated sanctions provide responses to violations of governance rules that are proportional and escalating. Conflict resolution mechanisms enable disputes to be resolved within the trust without recourse to external authorities. Minimal recognition by external authorities means the trust does not depend on State licensing for its legitimacy. Nested enterprises permit the trust to federate with other trusts for collective bargaining while retaining internal governance autonomy.

\begin{proposition}[D2 satisfies data sovereignty]
\label{prop:d2}
An cooperative data trust (Definition~\ref{def:trust}) whose governance rules $R$ satisfy Ostrom's eight design principles realizes the data sovereignty requirements of the digital social contract, specifically meaningful control, withdrawal capacity, and collective sovereignty, through voluntary institutional mechanisms rather than regulatory enforcement.
\end{proposition}

\begin{proof}
Meaningful control: by Ostrom's collective-choice principle, every member participates in defining the rules that govern their data, and the trust's small scale relative to a State regulatory apparatus ensures that participation is substantive rather than nominal. Withdrawal capacity: the exit function $E$ guarantees unconditional withdrawal with full data portability, a stronger guarantee than the GDPR's right to erasure, which is subject to multiple exceptions and whose enforcement depends on regulatory authorities. Collective sovereignty: the trust bargains collectively on behalf of its members, aggregating individual preferences into collective terms through democratic processes specified by $R$, and the federation principle enables trusts to coordinate without centralizing governance. The absence of State licensing means the trust's sovereignty is not contingent on State recognition, and the graduated sanctions principle provides enforcement without coercion: sanctions escalate from warning through temporary restriction of participation to expulsion, but never to imprisonment or financial penalties imposed by external authorities. At every stage, the member's exit right remains intact, ensuring that the trust's governance is disciplined by the threat of departure rather than the threat of punishment.
\end{proof}

Tucker's critique of monopoly~\citep{tucker1897} reinforces the sovereignty layer by identifying the structural conditions under which data sovereignty is illusory. Tucker argued that four State-enforced monopolies, patents, land, money, and tariffs, prevent genuine voluntary exchange by creating artificial scarcities that distort bargaining power. Applied to data governance, the relevant monopolies are intellectual property (which locks users into proprietary platforms), infrastructure monopoly (which concentrates data storage in corporate clouds), and monetary monopoly (which denominates data value in State currencies whose supply is controlled by central banks). The sovereignty layer therefore requires that the data trust operate on open protocols, store data on community-governed infrastructure (D1), and transact in mutual credit or cooperative currency (D6). Without these structural conditions, data sovereignty remains nominal: a member who can withdraw their data but has nowhere to take it except another corporate platform has not meaningfully exited.

\section{D3. Contextual Integrity: Federated Mutual-Aid Protocols}
\label{sec:contextual_layer}

The third dimension of the digital social contract addresses social cohesion: the maintenance of digital commons and cultural sustainability, grounded in the principle that data governance must respect the contextual norms governing different social domains~\citep{alvarezpallete2026}. The concept of contextual integrity, as developed by Nissenbaum~\citep{nissenbaum2004}, provides the philosophical foundation: data flows are appropriate when they conform to the informational norms of the context in which they occur, and violations arise when data crosses contextual boundaries without respecting the norms of either the originating or the receiving context.

The digital social contract proposes to enforce contextual integrity through purpose limitation requirements, mandatory impact assessments, and legal remedies for boundary violations. The structural difficulty, as noted in Section~\ref{sec:foundations}, is that the sovereign State is itself a persistent violator of contextual integrity, routinely accessing data across contextual boundaries in the name of national security, public order, and fiscal administration. A more fundamental difficulty is that contextual integrity, as a normative principle, requires governance structures that are themselves contextually embedded: a unitary regulator applying uniform rules across all contexts cannot respect contextual variation, because the specificity of context is precisely what a unitary jurisdiction abstracts away.

The commons-based alternative realizes contextual integrity through structural pluralism. Instead of a single regulatory authority applying context-sensitive rules from above, each data trust governs a specific social context and defines the informational norms appropriate to that context. Health data is governed by a health data trust whose members include patients and care providers; educational data is governed by an educational data trust; commercial data by a commercial trust; creative data by a cultural trust. Cross-contextual data sharing is not prohibited but must be negotiated between trusts through federated protocols that require bilateral consent: the originating trust must authorize the release under its own contextual norms, and the receiving trust must agree to handle the data under norms compatible with its own context.

\begin{definition}[Contextual federation protocol]
\label{def:federation}
A \emph{contextual federation protocol} is a bilateral agreement $\mathcal{F}_{ij}$ between data trusts $\mathcal{T}_i$ and $\mathcal{T}_j$ governing contexts $c_i$ and $c_j$ respectively, specifying:
\begin{align}
&\text{(i)}\quad \text{The data types eligible for cross-contextual sharing,} \\
&\text{(ii)}\quad \text{The purposes for which shared data may be processed in } c_j, \\
&\text{(iii)}\quad \text{The norms of } c_i \text{ that must be preserved in } c_j, \\
&\text{(iv)}\quad \text{Revocation conditions under which either trust may terminate sharing.}
\end{align}
The protocol requires the affirmative consent of both trusts' assemblies, and any member of either trust may invoke the revocation conditions.
\end{definition}

This federated structure realizes contextual integrity not as a regulatory constraint but as an architectural property: the governance system is composed of context-specific trusts, and cross-contextual data flows require bilateral negotiation, making boundary violations structurally difficult rather than merely illegal. The mutual-aid dimension, drawn from Kropotkin~\citep{kropotkin1902}, enters through the motivation for cross-contextual sharing: trusts share data not for commercial extraction but for mutual benefit, as when a health trust shares anonymized epidemiological data with a public-health trust to enable community health planning, or when an educational trust shares learning analytics with a pedagogical research trust to improve teaching methods. The norm is reciprocity: each trust contributes to and benefits from the federation, and free-riding is constrained by the graduated sanctions specified in each trust's Ostromian governance rules.

\begin{proposition}[D3 satisfies contextual integrity]
\label{prop:d3}
A system of context-specific data trusts connected by contextual federation protocols (Definition~\ref{def:federation}) satisfies contextual integrity structurally: boundary violations require the affirmative consent of both originating and receiving trusts' assemblies, and any member may invoke revocation, making violations detectable, reversible, and costly to the violating party.
\end{proposition}

\begin{proof}
A boundary violation, defined as a data flow that crosses contextual boundaries without respecting the norms of either context, requires that data move from $\mathcal{T}_i$ to $\mathcal{T}_j$ in a manner inconsistent with $\mathcal{F}_{ij}$. If no federation protocol exists, there is no authorized channel for inter-trust data flow, so any cross-contextual flow is a manifest violation rather than a sanctioned transfer. We do not claim exfiltration is technically impossible, since a defecting trust may always attempt to leak data outside the protocol; we claim that such leakage is detectable as unauthorized and triggers the revocation and reputational consequences described below, rather than passing as legitimate sharing. If a protocol exists, conditions (i)--(iv) constrain the flow, and both assemblies must have consented. A violation of the protocol by one trust (sharing data for unauthorized purposes, failing to preserve originating norms) triggers the revocation conditions, which any member of either trust may invoke. Revocation terminates the data flow and requires the receiving trust to delete or return shared data. The cost of violation is therefore immediate (loss of access to the federation) and reputational (other trusts will be reluctant to federate with a trust known to have violated protocols). Since the federation is voluntary and bilateral, no unitary authority can override contextual norms by fiat, and since any member can invoke revocation, violations cannot be concealed from the membership. The structure thus satisfies contextual integrity through architectural design rather than regulatory enforcement.
\end{proof}

Social cohesion is realized through the network of mutual-aid relationships that federated data sharing creates. Unlike a regulatory framework in which social cohesion is an aspirational goal to be pursued through policy, the federated trust architecture generates cohesion as a byproduct of its operation: trusts that cooperate effectively develop relationships of mutual trust and reciprocity that strengthen the social fabric, while trusts that defect are excluded from the federation, suffering the loss of the benefits that cooperation provides. This is the Kropotkinian insight applied to data governance: cooperation is not an altruistic sacrifice imposed from above but a rational strategy that emerges from the structure of mutual dependence.

\section{D4. Fiduciary: Voluntary Commitment and Reputational Accountability}
\label{sec:fiduciary_layer}

The fourth dimension of the digital social contract addresses automation limits: the requirement for meaningful human control over algorithmic decision-making, including the right to explanation, the right to contest, and the right to human review of automated decisions~\citep{alvarezpallete2026}. The contract operationalizes these requirements through fiduciary duties: data controllers are to assume duties of care, loyalty, and disclosure toward data subjects, with enforcement through the legal system. The fiduciary standard, as the contract argues, imposes higher obligations than consumer protection law because fiduciary duties cannot be disclaimed through terms of service and are imposed by law rather than negotiated between parties.

The commons-based critique of this arrangement is twofold. First, fiduciary law in its existing form is a product of the same legal system that created the corporate governance norms under which data extraction flourishes. The fiduciary duties of corporate directors run to shareholders, not to data subjects; the extension of fiduciary obligations to data controllers represents a reform within a legal framework whose structural logic favours capital. Second, and more fundamentally, the enforcement of fiduciary duties through the legal system reproduces the coercive dynamic that the non-instrumentalization thesis (Definition~\ref{ax:noninst}) identifies as problematic: a data controller who violates fiduciary duties faces fines, injunctions, and ultimately the coercive apparatus of the State, creating an institutional relationship in which compliance is sustained by punishment rather than by the mutual benefit that voluntary commitment generates.

The commons-based alternative does not abolish statutory fiduciary duty but adds to it a voluntary fiduciary commitment enforced by reputational accountability and exit rights, a commitment that binds whether or not the law also binds and that remains effective in the many settings where statutory enforcement is slow, captured, or absent. The conceptual basis is drawn from Errico Malatesta's pragmatic anarcho-communism~\citep{malatesta1891}, which insists that social cooperation requires trust but that trust must be earned through demonstrated conduct rather than imposed by legal compulsion. Malatesta's key insight is that voluntary agreements, freely entered into and freely terminable, generate more reliable cooperation than coerced compliance, because voluntary agreements are sustained by mutual benefit whereas coerced compliance is sustained by the threat of punishment and dissolves the moment the threat is removed.

\begin{definition}[Voluntary fiduciary commitment]
\label{def:fiduciary}
A \emph{voluntary fiduciary commitment} is a public declaration $\mathcal{C}(s, \mathcal{T})$ by a service provider $s$ to a data trust $\mathcal{T}$, specifying:
\begin{align}
&\text{(i)}\quad \text{Duty of loyalty: } s \text{ will not use members' data for purposes not authorized by } \mathcal{T}, \\
&\text{(ii)}\quad \text{Duty of care: } s \text{ will implement security and minimization standards specified by } \mathcal{T}, \\
&\text{(iii)}\quad \text{Duty of disclosure: } s \text{ will provide } \mathcal{T} \text{ with auditable records of all data processing,} \\
&\text{(iv)}\quad \text{Automation limits: no algorithmic decision affecting a member will be made without} \\
&\qquad\quad \text{the member's right to explanation, contest, and human review.}
\end{align}
The commitment is enforced by the trust's right to terminate the relationship and by the reputational consequences of violation in the federation.
\end{definition}

The enforcement mechanism operates through three channels. The first is exit: if a service provider violates its fiduciary commitment, the data trust terminates the relationship, and the provider loses access to the trust's data and the revenue or benefit that access provided. The second is reputation: violations are reported to the federation, and other trusts can condition their own engagement with the provider on its track record. In a federated ecosystem of many trusts, a provider that has violated commitments to one trust will find it difficult to establish relationships with others. The third is substitutability: the open-protocol infrastructure of the Dignity Stack (D1, D2) ensures that providers are substitutable, so that no trust is locked into a relationship with a provider whose fiduciary conduct is inadequate.

\begin{proposition}[D4 satisfies automation limits]
\label{prop:d4}
A system of voluntary fiduciary commitments (Definition~\ref{def:fiduciary}) in which service providers are substitutable, reputation is transparent across the federation, and data trusts retain unconditional exit rights satisfies the automation limits requirements of the digital social contract, specifically the rights to explanation, contest, and human review, without recourse to State-imposed fiduciary duties.
\end{proposition}

\begin{proof}
The right to explanation is specified in condition (iv) of the fiduciary commitment and is verifiable through the auditable records required by condition (iii). A member who receives an algorithmic decision without explanation can verify the violation by requesting the audit records; if the provider fails to produce them, the trust invokes its exit right. The right to contest is realized through the trust's governance structures: a member brings the contested decision to the assembly, which evaluates it against the fiduciary commitment and, if the commitment has been violated, terminates the relationship with the provider. The right to human review is specified in condition (iv) and enforced by the same mechanism. The substitutability of providers ensures that the trust's exit right is effective rather than nominal: because the system runs on open protocols and community-governed infrastructure, switching to a different provider does not require migrating to a different ecosystem. The reputational mechanism ensures that providers who violate commitments face consequences beyond the loss of a single trust's business, creating incentives for compliance that are proportional to the provider's interest in maintaining relationships across the federation.
\end{proof}

The Malatestian insight at work here is that voluntary commitment, reinforced by reputation and exit, generates more reliable compliance than legal compulsion. A data controller whose fiduciary conduct is motivated by the desire to avoid fines will calibrate its compliance to the expected cost of enforcement, investing in compliance only to the point where the marginal cost equals the expected fine discounted by the probability of detection. A service provider whose fiduciary conduct is motivated by the desire to maintain reputation in a federation of trusts faces a different calculus: the cost of violation is not a probabilistic fine but the certain loss of relationships, a consequence that is both more immediate and more proportional to the severity of the violation.

\section{D5. Participatory Governance: Nested Assemblies Without a Centre}
\label{sec:governance_layer}

The fifth dimension of the digital social contract is political legitimacy: the requirement that data governance rules be determined through democratic processes that include data subjects, with particular attention to the representation of marginalized communities~\citep{alvarezpallete2026}. The contract proposes multi-stakeholder governance bodies including data subjects, civil society, technical experts, and marginalized community representatives, with platform executives excluded for conflicts of interest. The difficulty with this proposal is not its ambition but its institutional form: multi-stakeholder governance bodies, as actually constituted in telecommunications, finance, and environmental regulation, routinely reproduce the power asymmetries they are designed to correct. The selection of ``representatives,'' the framing of agendas, the production of ``consensus'' documents, and the translation of recommendations into binding rules are all processes vulnerable to capture by organized interests with greater resources, expertise, and persistence than the diffuse public they claim to represent.

The commons-based alternative draws on Bookchin's libertarian municipalism~\citep{bookchin1982,bookchin2015} and Bakunin's federalist collectivism~\citep{bakunin1873} to construct a participatory governance architecture that is both radically democratic and scalable. The basic unit is the assembly: the face-to-face deliberative body of a data trust's members, meeting regularly to decide governance questions by consensus or qualified majority. Assemblies are small enough for genuine deliberation, following Bookchin's principle that democracy requires the physical co-presence of deliberating citizens and that assemblies exceeding a few hundred members tend toward representative delegation, which is to say, toward the dissolution of direct democracy into its negation.

For governance questions that transcend the individual trust, whether questions of federation policy, cross-contextual data sharing norms, or responses to external threats, assemblies send mandated delegates to a council of the federation. The delegates are mandated, not representative: they carry the specific decisions of their assemblies and have no authority to negotiate or compromise beyond their mandates. If the council confronts a question on which the delegates' mandates are insufficient, the question is referred back to the assemblies. This is Bakunin's federalism~\citep{bakunin1873}: authority flows from the base to the federation, never from the federation to the base, and delegates are recallable at any time by the assemblies that sent them.

\begin{definition}[Nested assembly governance]
\label{def:governance}
A \emph{nested assembly governance structure} is a tuple $\mathcal{A} = (A_1, \ldots, A_n, \mathcal{C}, \delta, \rho)$ where each $A_k$ is a local assembly of a data trust, $\mathcal{C}$ is a federation council, $\delta: A_k \to \mathcal{C}$ is the delegation function assigning mandated delegates from assemblies to the council, and $\rho: A_k \times \delta(A_k) \to \{0,1\}$ is the recall function permitting any assembly to recall its delegate at any time. The governance structure satisfies:
\begin{align}
&\text{(i)}\quad \text{Every governance decision affecting members of } A_k \text{ is made by } A_k, \text{ not by } \mathcal{C}. \\
&\text{(ii)}\quad \mathcal{C} \text{ may only coordinate between assemblies, not legislate over them.} \\
&\text{(iii)}\quad \text{Any } A_k \text{ may withdraw from the federation without forfeiting internal governance.}
\end{align}
\end{definition}

The political legitimacy of this structure derives not from the State's endorsement but from the direct participation of every person in the governance decisions that affect them. Where the digital social contract's multi-stakeholder bodies achieve legitimacy through the inclusion of ``representatives'' who stand for constituencies they cannot meaningfully consult, the nested assembly achieves legitimacy through direct participation at the local level and mandated delegation at the federation level. The difference is not one of degree but of kind: representation is a mechanism for managing the impossibility of direct participation at scale, while mandated delegation preserves direct participation by refusing to aggregate decisions beyond the scale at which face-to-face deliberation is possible.

\begin{proposition}[D5 satisfies political legitimacy]
\label{prop:d5}
A nested assembly governance structure (Definition~\ref{def:governance}) satisfies the political legitimacy requirements of the digital social contract, specifically democratic participation, representation of marginalized communities, and accountability of governance to data subjects, through direct participatory mechanisms rather than representative or multi-stakeholder bodies.
\end{proposition}

\begin{proof}
Democratic participation: every member of every trust participates directly in the assembly that governs their data, without delegation to representatives. This is a strictly stronger form of participation than the multi-stakeholder governance proposed by the digital social contract. Representation of marginalized communities: marginalized individuals participate directly in their assemblies on equal terms with all other members. Moreover, the cooperative data trust is constituted voluntarily, which means that marginalized communities may form their own trusts governed by their own norms, without depending on inclusion in institutions designed and controlled by others. Indigenous data sovereignty, for instance, is realized not through a seat on a multi-stakeholder board but through a self-governing data trust constituted by and for the indigenous community. Accountability: delegates to the federation council are mandated and recallable, creating a shorter and more direct chain of accountability than any representative or regulatory body. If a delegate exceeds their mandate, the assembly recalls them immediately; the lag between violation and correction is measured in days, not in the years that characterize regulatory proceedings. The withdrawal right~(iii) ensures that no trust is bound to governance decisions it did not consent to, realizing the principle of political self-determination at the institutional level.
\end{proof}

\section{D6. Economic Justice: Mutualist Exchange and Cooperative Surplus}
\label{sec:economic_layer}

The sixth dimension of the digital social contract is economic justice: the requirement that data subjects receive fair compensation for the value their data generates, realized through data dividends, cooperative ownership, reciprocal access rights, and alternative economic models that counteract the extractive logic of platform capitalism~\citep{alvarezpallete2026}. The contract proposes several mechanisms: a data dividend funded by a tax on data-derived profits, cooperative data models in which members collectively own data infrastructure, data trusts with revenue sharing, public data commons, and reciprocal data access.

Proudhon's mutualism~\citep{proudhon1840,proudhon1851} provides the commons-based organizational protocol for this layer. Proudhon's economic thought centres on two principles: the labour theory of exchange value, according to which the just price of a good or service is determined by the labour required to produce it, and the mutual credit system, which replaces State-issued currency and bank-intermediated credit with a reciprocal credit system in which participants extend credit to one another without interest, eliminating the rent extracted by financial intermediaries. Applied to data economics, mutualism implies that the value of data should flow to the persons whose labour (in the extended sense of lived experience, creative production, and social interaction) generated it, and that the medium of exchange through which data value circulates should be a mutual credit system rather than a State currency whose supply is controlled by central banks and whose distribution is intermediated by profit-seeking financial institutions.

\begin{definition}[Mutualist data economy]
\label{def:economy}
A \emph{mutualist data economy} is a system $\mathcal{E} = (\mathcal{T}_1, \ldots, \mathcal{T}_n, \mu, \sigma)$ where each $\mathcal{T}_k$ is a data trust, $\mu$ is a mutual credit system enabling inter-trust exchange without interest or financial intermediation, and $\sigma: \mathcal{T}_k \to \mathbb{R}^{|M_k|}$ is a surplus distribution function that allocates each trust's surplus (the difference between the value generated by collective data stewardship and the costs of trust operation) to its members according to rules determined by the assembly.
\end{definition}

The mutual credit system $\mu$ is the institutional mechanism that prevents the re-enclosure of data value by financial intermediaries. In the digital social contract's proposed data dividend model, data-derived value is first captured by corporations as profit, then taxed by the State, and finally distributed to data subjects as a transfer payment. This circuitous route introduces three intermediaries, each of which extracts rent: the corporation retains the majority of the value as profit after tax, the State retains a portion for administrative costs, and the financial system extracts fees for the transfer. The mutualist alternative eliminates all three intermediaries: data trusts transact directly with service providers and with one another through the mutual credit system, and the surplus is distributed to members according to rules they themselves determine.

The surplus distribution function $\sigma$ is not a data dividend in the sense of a universal payment determined by the State. It is a cooperative surplus determined by each trust's assembly, which may choose to distribute it equally, to weight it by contribution, to reinvest it in infrastructure, or to direct it toward solidarity funds that support trusts with fewer resources. The key difference from the digital social contract's proposal is that the distribution is determined democratically by the persons whose data generated the value, not by a State tax-and-transfer mechanism whose rules are set by legislatures subject to corporate lobbying.

\begin{proposition}[D6 satisfies economic justice]
\label{prop:d6}
A mutualist data economy (Definition~\ref{def:economy}) satisfies the economic justice requirements of the digital social contract, specifically fair compensation, alternative economic models, and counteraction of extractive logic, through voluntary cooperative mechanisms rather than State taxation and redistribution.
\end{proposition}

\begin{proof}
Fair compensation: the surplus distribution function $\sigma$ ensures that value flows to data subjects, and the assembly governance ensures that the distribution rules are determined by the data subjects themselves, not by external authorities. The mutualist structure eliminates the intermediary rents that reduce the share of value reaching data subjects in the tax-and-dividend model. Alternative economic models: the mutual credit system $\mu$ constitutes a concrete alternative to the extractive logic of platform capitalism, replacing the extraction-profit-tax-dividend chain with a direct cooperative exchange. Counteraction of extractive logic: the data trust's governance ensures that data is not treated as a commodity to be sold for maximum revenue but as a commons whose benefits are distributed according to collectively determined rules. The absence of external shareholders ensures that no residual claim on the trust's surplus exists beyond the members' collectively determined distribution. Reciprocal access: the federation protocol (D3) enables inter-trust data sharing on terms negotiated by the trusts' assemblies, realizing reciprocal access without the regulatory mandates the digital social contract proposes.
\end{proof}

\section{Integration: The Dignity Stack as a Coherent Architecture}
\label{sec:integration}

The six layers of the Dignity Stack are not independent modules but an integrated architecture in which each layer depends on and reinforces the others. This section specifies the inter-layer dependencies and demonstrates systemic coherence.

\begin{definition}[Inter-layer dependency]
\label{def:dependency}
The inter-layer dependency function $\Phi$ of the Dignity Stack specifies the following relations:
\begin{align}
\Phi(D_1, D_2) &= 1: \text{ D2 requires D1 for material independence.} \\
\Phi(D_2, D_3) &= 1: \text{ D3 requires D2 for trust constitution.} \\
\Phi(D_2, D_4) &= 1: \text{ D4 requires D2 for effective exit rights.} \\
\Phi(D_3, D_5) &= 1: \text{ D5 requires D3 for inter-trust coordination.} \\
\Phi(D_1, D_6) &= 1: \text{ D6 requires D1 for non-corporate infrastructure.} \\
\Phi(D_5, D_6) &= 1: \text{ D6 requires D5 for democratic surplus rules.}
\end{align}
\end{definition}

Each dependency is substantive. D2 requires D1 because a data trust whose infrastructure is hosted on corporate cloud servers is not truly sovereign: the cloud provider can access the data, surveil usage patterns, or terminate service unilaterally. Only community-governed infrastructure guarantees the material independence that sovereignty presupposes. D3 requires D2 because contextual integrity is realized through federation between sovereign trusts; trusts that are not sovereign cannot credibly commit to contextual norms. D4 requires D2 because the exit right that enforces fiduciary commitments is effective only if the exiting member has somewhere to go; without sovereign infrastructure and portable data formats, exit is a formal right without material substance. D5 requires D3 because inter-trust governance coordination is conducted through the federation protocols that D3 establishes. D6 requires both D1 and D5: D1 provides the non-corporate infrastructure on which the mutual credit system operates, and D5 provides the democratic governance through which surplus distribution is determined.

\begin{theorem}[Systemic coherence]
\label{thm:coherence}
A Dignity Stack $\mathcal{D}$ that satisfies the normative compliance condition (Definition~\ref{def:compliance}), i.e., each layer satisfies its verification predicate and all inter-layer dependencies are met, implements the full digital social contract without recourse to State regulatory enforcement while satisfying the non-instrumentalization thesis (Definition~\ref{ax:noninst}), the relational personhood thesis (Definition~\ref{ax:relational}), and the structural contextual integrity thesis (Definition~\ref{ax:context}).
\end{theorem}

\begin{proof}
By Propositions~\ref{prop:d1}--\ref{prop:d6}, each layer of the Dignity Stack satisfies the normative requirements of the corresponding dimension of the digital social contract through voluntary institutional mechanisms. It remains to verify that the integrated architecture satisfies the three foundational axioms. The non-instrumentalization thesis (Definition~\ref{ax:noninst}): no layer of the Dignity Stack employs coercion. The governance mechanisms are consent (D1 assemblies, D2 trust membership), exit (D2 exit function, D4 termination rights), reputation (D4 federation-wide reporting), and mandated delegation (D5 nested assemblies). All institutional relationships are constituted through ongoing, meaningful, and revocable consent, satisfying the thesis more fully than any architecture that includes coercive authority. The relational personhood thesis (Definition~\ref{ax:relational}): the entire architecture is constituted through voluntary associations, from the local data trust through the federation to the mutualist economy, in which persons relate to one another as members of a community of mutual recognition rather than as subjects of an authority. The data trust is not a bureaucratic intermediary but a self-governing community in which data stewardship is an expression of mutual care. The structural contextual integrity thesis (Definition~\ref{ax:context}): no institution in the Dignity Stack has jurisdiction across all contexts. Each data trust governs a specific context, and cross-contextual data flows require bilateral negotiation through the federation protocol (D3). The nested assembly governance (D5) coordinates between contexts but cannot override contextual norms, because assemblies retain sovereignty over their internal governance and the right to withdraw from the federation. The inter-layer dependencies (Definition~\ref{def:dependency}) ensure that each layer's institutional mechanisms are supported by the layers beneath it, preventing the systemic failure that would occur if, say, a data trust claimed sovereignty while operating on corporate infrastructure.
\end{proof}

The architecture admits one further question that the layer-by-layer construction does not yet answer and that the objections of Section~\ref{sec:objections} press hardest. If the commons may be funded by firms and States, what prevents the largest funders from converting their contributions into control, reproducing within the commons the very concentration of power the digital social contract exists to resist? The answer is partly structural and partly empirical, and the boundary between the two must be drawn with care, because it is here that an optimistic reading of the architecture is most tempting and least warranted.

It is necessary first to distinguish two forms of capture, since they require different defences and conflating them is the characteristic error of optimistic commons design. \emph{Formal} capture is the acquisition of governance entitlements: additional votes, a veto, or a residual claim on the commons' surplus or direction. \emph{Structural} capture is the acquisition of leverage without any formal entitlement, through the dependence of the commons on a resource that a contributor remains free to withdraw. The governance rules of the architecture bear directly on the first and only indirectly on the second.

\begin{definition}[Capital--governance decoupling]
\label{def:decoupling}
A commons exhibits \emph{capital--governance decoupling} when the right to participate in its governance is held equally by its participating members and is not a function of the capital, compute, or data that any party contributes to it. Contributions are of two kinds, which the definition keeps distinct. An \emph{endowment} contribution transfers an asset to the commons once and for all and, once transferred, creates no continuing dependence. A \emph{supply} contribution provisions an ongoing operational input, such as compute or energy, that the contributor continues to control at its source and may cease to provide. In neither case does a contribution confer additional vote, veto, or residual claim on the commons' surplus or direction beyond terms set by the assembly; but a supply contribution creates a dependence whose governance significance the formal rules do not by themselves dissolve.
\end{definition}

\begin{proposition}[Capitalization without formal capture]
\label{prop:capture}
A Dignity Stack that satisfies capital--governance decoupling (Definition~\ref{def:decoupling}) absorbs capital and compute from firms and States without transferring \emph{formal} governance entitlements to those contributors, provided three conditions hold: (i) governance proceeds by one participant, one voice, with mandated and recallable delegation (D5); (ii) contributed capital earns no rent and no residual claim on surplus beyond contractually agreed repayment, in the mutualist sense (D6); and (iii) the assembly's own rules forbid it to accept any contribution conditioned on a relaxation of the dignity-by-design, consent, contextual-integrity, or economic-justice requirements of the digital social contract.
\end{proposition}

\begin{proof}
Conditions (i) and (ii) define governance rights and surplus claims independently of contribution, so a contributor, whatever the magnitude of its capital, acquires the same single voice as any other participant and no claim on surplus or direction; the vote channel and the residual-claimant channel are closed by construction. Condition (iii) closes the channel of conditional acceptance: a contribution offered in exchange for a governance concession is one the assembly's constitutive rules forbid it to accept. What the three conditions establish is therefore exactly their conjunction, that no formal governance entitlement is transferred through votes, through surplus, or through accepted conditions. They do not establish that the assembly will in fact apply rule (iii), which presupposes an assembly not already captured; the proposition states what decoupling forecloses, not that an assembly remains disposed to enforce it, and that residual is taken up below. The conditions are individually necessary, since dropping any one reopens its channel, and jointly sufficient against formal capture alone. Structural capture lies outside their scope, as the next paragraphs make explicit.
\end{proof}

The confinement to formal capture is deliberate, because the channel that the financing of fabrication plants and gigawatt data centres actually opens is structural. A contributor that supplies a majority of the commons' ongoing operational compute need acquire no vote and demand no surplus to exercise control; it need only retain, as Definition~\ref{def:decoupling} grants every supply contributor, the ability to cease supplying. The threat of withdrawal is a control lever that conditions (i) through (iii) leave untouched, and it is precisely the lever that the historical record of resource-asymmetric cooperatives and foundations shows to be decisive. The backstop on which formal-capture resistance relies, the exit and federation-exclusion mechanisms of D2 and D3, is moreover weakest exactly where structural capture is strongest: a commons whose compute is a single funder's gigawatt data centre has, by the network-effect and sunk-cost argument we concede in Section~\ref{sec:objections}, nowhere to exit to. The convivial threshold of Definition~\ref{def:convivial} was the original safeguard against this dependence; the present section does not abolish it but relocates it, as we make explicit both here and in our reply to the material-base objection.

\begin{proposition}[Resistance to structural capture is bounded by polycentric supply]
\label{prop:structural}
Structural capture of a Dignity Stack is resisted to the degree that its essential operational supply is polycentric: that no single contributor provisions more than a bounded share of any essential ongoing input, and that inputs are substitutable, so that the withdrawal of any single supplier is survivable through redundancy or replacement. Resistance degrades continuously as supply concentrates and is absent in the limit of monopoly supply, whatever the formal governance rules.
\end{proposition}

This proposition is bounded by design, and the bound is the honest content of the architecture's claim about the lower layers. The Dignity Stack does not establish that frontier-scale infrastructure can be funded by a single hyperscaler or State without capture; it establishes the opposite, that single-source funding of essential infrastructure is capture in all but name, and that the commons resists it only insofar as it can multi-source its supply. The material-base concession of earlier framings is therefore displaced rather than dissolved: the lower layers are fundable by firms and States without capture not at arbitrary concentration but only under a polycentric-supply condition that the present structure of the chip and data-centre markets makes demanding to satisfy. Open hardware, pooled and substitutable compute, and a plurality of competing suppliers are not peripheral conveniences but the precondition on which the whole claim rests, and where they are absent the claim fails.

It is in this bounded sense, and only in it, that the Dignity Stack functions as a shared civic battery rather than a private asset. A battery charged from a single source that can be switched off is controlled by that source regardless of who nominally owns it; a battery charged from many substitutable sources is not. The governance rules of Proposition~\ref{prop:capture} keep capital from buying direction outright, but it is the polycentric-supply condition of Proposition~\ref{prop:structural} that keeps capital from commanding direction through the threat of withdrawal, and only the two together make the metaphor honest.

We state the residual plainly, because two distinct risks remain even when both propositions are satisfied. The first is that condition (iii) requires an assembly disposed to enforce it, and a sufficiently dependent assembly may not be: capital--governance decoupling is a design commitment, not a law of nature, and the classical thesis of cooperative degeneration records the repeated erosion of exactly such commitments under sustained resource asymmetry, as large contributors convert provisioning into informal influence, staff dependence, and the slow rewriting of rules in their favour. The second is that the polycentric-supply condition may be unattainable at the fabrication layer, where the market is at present close to a monopoly, so that the bounded claim of Proposition~\ref{prop:structural}, though valid, may have no satisfiable instance at that layer for the foreseeable future. We do not argue these risks away. We claim only that locating them precisely, the first at the disposition of the assembly and the second at the concentration of supply, and equipping the commons with one-voice governance, no-rent rules, conditional-acceptance rules, multi-sourcing, and exit, is a stronger and more honest position than either conceding the lower layers to the incumbents outright or pretending that funding at scale carries no risk of capture.

\section{Objections and Responses}
\label{sec:objections}

The voluntary governance architecture proposed in this paper invites several serious objections. We address eight, including those we consider most damaging, and respond to each with the honesty that rigorous philosophical argument demands.

The first objection is the \emph{enforcement problem against powerful adversaries}: without the State's coercive apparatus, what prevents actors with State-level resources, whether multinational corporations with hundred-billion-dollar revenues or national intelligence agencies, from systematically violating the Dignity Stack's norms? This is the strongest version of the enforcement objection, and we do not underestimate it. The objection has two forms that must be separated. In the first, a powerful actor contributes capital or compute to the commons and seeks to convert its contribution into control; this form is answered by the capital--governance decoupling of Proposition~\ref{prop:capture}, which denies a funder any governance purchase that its capital might otherwise buy. In the second, a powerful actor external to the commons violates its norms by force or fiat; it is this second form that the present objection presses, and we address it directly. Ostrom's empirical research demonstrates that commons governed by communities can be stable and effectively maintained~\citep{ostrom1990}, but Ostrom's cases, fisheries, forests, and irrigation systems, involve defectors who lack the power to overwhelm the governance structure. A corporation that decides to infiltrate a data trust, or a State agency that demands access to trust data under national security authority, poses a categorically different challenge. Our response is twofold. First, the enforcement mechanisms of the Dignity Stack, which comprise exit, reputation, graduated sanctions, and federation exclusion, do impose real costs on violators and are effective against the range of ordinary violations that constitute the bulk of data governance failures. Second, we acknowledge that the Dignity Stack cannot, by itself, resist an adversary that is willing to use the coercive power of the State to override its governance. This is a genuine limitation, and one the commons shares with every private owner of the same infrastructure, none of whom can resist a determined State either; the commons does not seek to evade the State but to coexist with it. It means that the Dignity Stack operates most effectively either within a legal environment that recognizes the right of voluntary associations to govern their own affairs, which is Ostrom's seventh design principle (``minimal recognition of rights to organize''), or in domains where State power is practically limited. We do not claim that voluntary governance is sufficient against all adversaries; we claim that it is more effective than regulatory enforcement against the vast majority of data governance violations, while being normatively superior on the dimensions of dignity, consent, and contextual integrity. The statist alternative, for comparison, is also unable to enforce data rights against the State itself: intelligence agencies operate outside civilian data protection regimes by design, and GDPR fines, even the record \texteuro{}1.2 billion imposed on Meta, represent approximately 4\% of annual revenue, a cost of doing business rather than a meaningful deterrent.

The second objection is the \emph{scalability problem}: can voluntary governance scale to the level required for a comprehensive data governance system? The Dignity Stack does not require a single governance unit to scale to millions of members; it requires many small governance units to federate. Federation is a proven scalability mechanism: the Internet itself is a federated architecture in which autonomous systems coordinate through bilateral agreements (BGP peering). The Fediverse (Mastodon, Lemmy, PeerTube) demonstrates that social media can operate on a federated model with millions of users and thousands of independently governed instances. However, intellectual honesty requires acknowledging that the Fediverse also illustrates governance challenges: documented difficulties with cross-instance content moderation, loss of social capital during migration between instances, and undocumented governance norms that create barriers to entry. These are real problems, not fatal ones: they indicate that federated governance requires ongoing institutional learning and adaptation, not that it is infeasible. The cooperative movement, comprising over 3 million cooperatives worldwide with over 1 billion members, demonstrates that voluntary governance at scale is an institutional reality, not a hypothesis.

The third objection is the \emph{defection and data irreversibility problem}: what prevents individuals or trusts from free-riding on the cooperative system, and what happens when data, once shared, cannot be ``returned''? The Ostromian governance framework addresses ordinary defection through clearly defined boundaries, monitoring, graduated sanctions, and federation exclusion. The more serious version of this objection, however, concerns the irreversibility of data in derived artifacts. If a trust uses members' data to train a machine learning model before a member exits, the exit function (Definition~\ref{def:trust}), which guarantees withdrawal of raw data with no copy retained, does not address the information that has been incorporated into model weights. Trained models are, in a meaningful sense, copies of their training data. We acknowledge this as a genuine limitation of the exit mechanism as currently defined. The honest response has three parts. First, this problem is not unique to the commons-based framework; it applies with equal force to the GDPR's right to erasure and to the DSC's proposed data trusts. No existing governance framework has solved the problem of ``unlearning'' trained models. Second, emerging technical work on machine unlearning and differential privacy offers partial mitigation, and the Dignity Stack's open-protocol infrastructure (D1) facilitates the adoption of such techniques. Third, the trust's governance rules ($R$) can specify ex ante constraints on model training, including requirements for differential privacy guarantees, model auditing, and restrictions on the retention of models trained on departed members' data. The limitation remains, but it is a shared limitation of all data governance frameworks, not a specific weakness of the commons-based alternative.

The fourth objection is the \emph{tension between individual sovereignty and collective data governance}: if data is an expression of individual personhood (as data personalism claims), how can it be governed collectively through data trusts? The tension is real but is not unique to the commons-based framework; it is inherent in the digital social contract itself, which simultaneously affirms individual data sovereignty and proposes collective governance through data trusts. The commons-based resolution is the exit function: the individual's sovereignty over their data is realized through the unconditional right to withdraw from the trust with a complete copy of their data, subject to the model-weights limitation acknowledged above. Collective governance is thus conditional on ongoing individual consent, and the individual retains ultimate authority over their data at all times. This is a stronger protection than the statist alternative offers: a citizen cannot withdraw from a State's data protection regime without physically relocating to another jurisdiction, a vastly more costly and disruptive exit than withdrawing from a data trust.

The fifth objection is the \emph{tyranny of structurelessness}: Jo Freeman's influential analysis~\citep{freeman1972} argues that the absence of formal authority structures does not eliminate power but merely renders it invisible, unaccountable, and concentrated in the hands of informal elites. Applied to data trusts, this objection warns that voluntary governance may produce informal hierarchies that are harder to challenge than formal ones. The objection is well-taken, and the Dignity Stack's design incorporates Freeman's own prescriptions: explicit governance rules ($R$ in Definition~\ref{def:trust}), transparent decision-making through assemblies, mandated delegation with recall (Definition~\ref{def:governance}), and the exit function as an individual safeguard against informal capture. The Dignity Stack is not structureless; it is structured through explicit, democratically determined rules rather than through imposed external authority. Freeman's critique targets groups that claim to have no structure while operating through invisible hierarchies; the Dignity Stack's formal governance architecture is designed precisely to prevent this pathology.

The sixth objection is the \emph{tragedy of the anticommons}: Michael Heller~\citep{heller1998} argues that excessive fragmentation of governance rights can produce gridlock and underuse, as multiple rights-holders each possess veto power over productive uses of a shared resource. Applied to the Dignity Stack, this objection warns that a proliferation of context-specific data trusts, each with the power to refuse cross-contextual sharing, might prevent beneficial uses of data that require aggregation across contexts, such as public health surveillance, climate modelling, or educational research. The objection identifies a genuine risk. The federated mutual-aid protocol (D3, Definition~\ref{def:federation}) is designed to mitigate it: trusts share data for mutual benefit through bilateral agreements, and the norm of reciprocity creates incentives for cooperation. The risk of anticommons gridlock is real but must be weighed against the risk of the alternative: a unitary authority with the power to override contextual norms may enable beneficial data aggregation, but it also enables the surveillance, profiling, and contextual boundary violations that the digital social contract exists to prevent. The Dignity Stack accepts slower, negotiated data sharing as the price of structural contextual integrity.

The seventh objection concerns the \emph{network effects of exit}: Hirschman~\citep{hirschman1970} himself warned that exit alone can destroy institutions, and that effective governance requires both exit and voice. In digital ecosystems, network effects raise the cost of exit: a data trust's value to its members depends on how many members it has and how much data it aggregates, so that exiting means losing access to collective bargaining power, shared analytical resources, and social connections. We accept this objection as identifying a real constraint on the exit mechanism. The Dignity Stack's response is structural: by requiring open protocols and community-governed infrastructure (D1, D2), it reduces switching costs relative to the corporate platforms where network effects are most severe. Moreover, the nested assembly governance (D5) provides voice mechanisms, democratic participation, mandated delegation, and recall, alongside exit. The Dignity Stack does not rely on exit alone; it combines exit with robust voice mechanisms. The remaining network-effect costs of exit are a genuine limitation, and we do not claim otherwise.

The eighth objection is the \emph{material base problem}: even granting that the governance protocols of the Dignity Stack are normatively superior, the lower layers of the very stack this paper invokes resist decentralization for reasons of capital and physics rather than ideology. \citet{huang2026} partitions the contemporary AI system into energy, chips, infrastructure, models, and applications, and the lower three of these layers exhibit economies of scale and lock-in that no voluntary community can plausibly overcome: leading-edge fabrication requires facilities whose cost runs to tens of billions of dollars, frontier training requires data centres drawing power on the order of gigawatts, and the dominant accelerator vendor's software ecosystem constitutes a moat sustained by millions of developers whose accumulated tooling raises the switching cost of any open alternative. The convivial threshold of Definition~\ref{def:convivial}, which requires that a system's energy demand not exceed what its community can self-provision and that its workings remain auditable by that community, appears to resolve this problem only by definitional fiat: it restricts the Dignity Stack to AI systems small enough to be community-powered and community-understood, thereby conceding the frontier of capability, precisely the systems that concentrate economic and geopolitical power, to the incumbents the architecture was meant to displace. On this reading the commons would govern the periphery while the commanding heights remained in corporate and State hands.

The objection draws its force from an assumption the commons does not make: that to govern a layer of the stack one must own and operate it oneself, so that a community unable to build a fabrication plant must cede the fabrication layer to whoever can. The battery model of Section~\ref{sec:integration} rejects exactly this identification of ownership with governance. A community need not fabricate its own silicon, nor raise the tens of billions a leading-edge plant costs, in order to govern its relation to silicon on commons terms. Under capital--governance decoupling (Proposition~\ref{prop:capture}), the capital for fabrication plants, gigawatt data centres, and frontier training runs can come from firms and States, precisely the actors with the means to supply it, while the governance of the resulting capacity, namely who may use it, under what consent and contextual-integrity norms, and how its surplus is distributed, remains with the commons. Capital purchases throughput; it does not purchase direction. The frontier of capability is on this account fundable without being captured, and the commanding heights need not remain in corporate and State hands merely because corporate and State capital built them.

This reframing changes the status of the convivial threshold of Definition~\ref{def:convivial}, and honesty requires distinguishing which of its three conditions the change reaches and which it does not. The epistemic and political conditions, the knowledge relation~\eqref{eq:knowledge} and the deliberation relation~\eqref{eq:dismantle}, are met at the level of the federation rather than the single community: a single community cannot by itself audit a frontier model or oversee a hyperscale data centre, but a federation can pool the technical competencies required and exercise mandated, recallable oversight (D5) over a system that exceeds any one of its members. To this extent the threshold is not a ceiling that confines the commons to small systems but a requirement that auditing and oversight capacity scale, at the federation level, with the system governed. The material condition is a different matter. The energy relation~\eqref{eq:energy} demands self-provisioning, and a gigawatt data centre funded and powered by a hyperscaler satisfies it for no community and no federation; such a system has governance sovereignty, in that the federation can audit and in principle decommission it, but not material sovereignty, in the strict sense of Definition~\ref{def:convivial} and the D1 to D2 dependency of Definition~\ref{def:dependency}, which holds that a provider able to terminate supply unilaterally compromises sovereignty. We therefore do not claim that the lower layers, as funded today, are sovereign in the full sense the convivial threshold defines. We claim the weaker and defensible thing: that governance sovereignty is attainable at federation scale now, and that material sovereignty is recovered only as the polycentric-supply condition of Proposition~\ref{prop:structural} is approached, through multi-sourced and substitutable provision. Where supply remains single-sourced, the convivial threshold is violated at condition~\eqref{eq:energy} and the sovereignty of that layer is partial, a fact we mark rather than reinterpret away.

Two considerations bound what remains. First, the convivial threshold retains a normative edge independent of capacity: much of the apparent necessity of frontier-scale compute reflects business models built on advertising, surveillance, and centralized capture rather than the intrinsic requirements of the tasks users actually need performed, so that the commons can decline to reproduce that scale where it is an artifact of extraction rather than of need. Second, the claim defended in this paper is falsifiable rather than self-sealing. It would be refuted if capital--governance decoupling could not in fact be sustained at fabrication and data-centre scale, that is, if funders of that magnitude proved able to convert their contributions into control despite the safeguards of Proposition~\ref{prop:capture}, or if commons-governed AI at the upper layers systematically failed to sustain itself where the legal environment permitted it under Ostrom's seventh design principle of minimal recognition of the right to organize~\citep{ostrom1990}. The honest residual is therefore not that the lower layers are unfundable, which the battery model answers, but that the decoupling on which their funding depends is, as Section~\ref{sec:integration} concedes, a governed commitment rather than a guarantee. We mark that boundary rather than conceal it.

\section{Conclusion}
\label{sec:conclusion}

The digital social contract proposed by Alvarez-Pallete et al.\ represents a significant philosophical advance: it grounds data governance in human dignity, identifies the categorical imperative and personalism as the relevant normative frameworks, and articulates six dimensions of governance that a dignity-centric digital social contract must address. This paper has argued that the normative content of that contract, while not strictly entailing any single institutional architecture, is more faithfully realized by voluntary, polycentric, commons-governed institutions than by the statist regulatory framework the contract presupposes. The three philosophical theses we have defended, concerning non-instrumentalization, relational personhood, and structural contextual integrity, do not derive mechanically from Kant, Mounier, or Nissenbaum; they draw those traditions in a commons-based and subsidiary direction, motivated by the empirical observation that actually existing States, like the private owners of the stack, are structurally implicated in the data extraction the contract seeks to prevent. In the case of Mounier the affinity is closest on the governance form, the commons supplying the kind of subsidiary mediation between person and State his personalism called for, though our economic layer extends past his acceptance of the market, a divergence we mark rather than elide. The Dignity Stack provides the constructive alternative: a six-layer governance architecture in which each dimension of the digital social contract is implemented by a commons-based organizational protocol drawn from the cooperative, mutualist, and libertarian-municipalist traditions, specified with sufficient precision to evaluate its operational viability.

The argument is not that voluntary governance of AI is easy, desirable because easy, or inevitable. It is that the values the digital social contract proclaims, specifically the irreducibility of human dignity, the requirement of meaningful consent, the prohibition on instrumentalization, the demand for economic justice, and the principle of contextual integrity, are more fully realized in institutions constituted by voluntary association than in institutions constituted by coercion, even when those coercive institutions are designed with the best of intentions. We have been candid about the limitations of the commons-based alternative: it cannot by itself resist adversaries with State-level coercive power, a limitation it shares with the private owners of the same infrastructure; its exit mechanism does not fully address the irreversibility of data in trained models; its federated structure risks anticommons gridlock; and its governance units face the classic challenges of voluntary association, including informal hierarchies and network-effect switching costs. The deepest of these, and the one that replaces the material-base concession of earlier framings, concerns the funding of the capital-intensive lower layers by firms and States. The governance rules of the architecture defeat formal capture, the acquisition of votes or surplus, but not structural capture, the leverage a dominant supplier holds through the threat of withdrawing what it provides; that second form is resisted only to the degree that operational supply is polycentric and substitutable, a condition the present concentration of the chip and data-centre markets makes hard to satisfy and perhaps, at the fabrication layer, for now impossible. Capital--governance decoupling is moreover a governed commitment rather than a guarantee, since it presupposes an assembly disposed to enforce it, and the classical pattern of cooperative degeneration shows such commitments eroding under sustained resource asymmetry. These are the principal risks the architecture must manage, and we neither resolve them here nor conceal them. These limitations are real, and we do not minimize them. But they must be weighed against the limitations of the statist alternative, which include regulatory capture, State complicity in surveillance, the structural power asymmetry between regulators and regulated corporations, and the inherent tension between a coercive enforcement mechanism and the dignity-centred values it is meant to protect.

The practical path from the present institutional landscape to the Dignity Stack is a question that exceeds the scope of this paper. What we have established is that the voluntary governance of AI data is not merely a utopian aspiration but a constructive possibility grounded in actually existing institutional models, from cooperative data trusts through municipal assemblies to federated protocols, and that it constitutes, on the philosophical terms the digital social contract itself articulates, a more faithful realization of the values it proclaims. The community-governed AI system that powers itself, hosts its own data, governs its own norms, chooses its own service providers, deliberates its own governance, and distributes its own surplus is not a utopian fantasy imposed on the digital social contract from outside. It is the digital social contract taken seriously.


\end{document}